\documentclass{elsarticle}
\usepackage[a4paper, top=2.5cm, bottom=2.5cm, left=2.5cm, right=2.5cm]{geometry}
\usepackage{natbib}
\setcitestyle{authoryear,round}
\makeatletter
\def\ps@pprintTitle{%
 \let\@oddhead\@empty
 \let\@evenhead\@empty
 \def\@oddfoot{}%
 \let\@evenfoot\@oddfoot}
\makeatother
\usepackage{amsmath}
\usepackage{graphicx}
\usepackage{ragged2e}
\usepackage{float}
\usepackage{todonotes}
\usepackage{placeins}
\usepackage{bigints}
\usepackage{booktabs}

\DeclareMathOperator*{\argmin}{arg\,min}
\usepackage{lscape} 

\usepackage{tabularx}
\usepackage{arydshln}
\usepackage{mathtools}
\usepackage{afterpage}
\usepackage{amsfonts}
\usepackage{algorithm,algpseudocode}
\usepackage{arydshln}
\usepackage[scr]{rsfso}
\usepackage{multirow}
\usepackage[flushleft]{threeparttable}
\usepackage{siunitx}
\usepackage{caption}
\usepackage{rotating}
\usepackage{amsfonts}
\usepackage{booktabs}
\usepackage{siunitx}
\usepackage{siunitx}
\usepackage{subcaption}
\usepackage{appendix}
\usepackage{amsthm}
\usepackage{amssymb,bm}
\usepackage{amsfonts}
\definecolor{applegreen}{rgb}{0.55, 0.71, 0.0}
\definecolor{ao(english)}{rgb}{0.0, 0.5, 0.0}
\usepackage{booktabs,caption,fixltx2e}
\usepackage[pdftex,colorlinks,
  citecolor=black]{hyperref}
\usepackage{setspace}
\usepackage{bbold}
\DeclareMathOperator{\vecop}{vec}
\newcommand{\norm}[1]{{\left\lVert#1\right\rVert}}

\usepackage{dsfont}
\usepackage{comment}
\usepackage{color}
\usepackage{graphicx}

\usepackage[normalem]{ulem}

\newcommand{\BS}[1]{\boldsymbol{#1}}

\newcommand{\bA}{\bm A}

\newcommand{\bB}{\bm B}

\newcommand{\bD}{\bm D}

\newcommand{\bE}{\bm E}
\newcommand{\be}{\bm e}

\newcommand{\bof}{\bm f}

\newcommand{\bH}{\bm H}

\newcommand{\bI}{\bm I}

\newcommand{\bL}{\bm L}

\newcommand{\bS}{\bm S}

\newcommand{\bU}{\bm U}
\newcommand{\bu}{\bm u}

\newcommand{\bv}{\bm v}
\newcommand{\bW}{\bm W}
\newcommand{\bw}{\bm w}
\newcommand{\bX}{\bm X}
\newcommand{\bx}{\bm x}
\newcommand{\bY}{\bm Y}

\newcommand{\E}{\mathbb{E}}

\newcommand{\balpha}{\bm \alpha}
\newcommand{\bbeta}{\bm \beta}
\newcommand{\bgamma}{\bm \gamma}

\newcommand{\boeta}{\bm \eta}

\newcommand{\blambda}{\bm \lambda}

\newcommand{\bxi}{\bm \xi}

\newcommand{\gt}{g(N,T,\zeta)}

\newcommand{\bGamma}{\bm \varGamma}

\newcommand{\bTheta}{\bm \varTheta}
\newcommand{\bLambda}{\bm \varLambda}

\newcommand{\bSigma}{\bm \varSigma}

\newcommand{\var}{\mbox{Var}}
\newcommand{\cov}{\mbox{Cov}}

\newcommand{\eps}{\varepsilon}
\newcommand{\R}{\mathds{R}}

\newcommand{\A}{\boldsymbol{\frak{A}}}

\newcommand{\Z}{\mathds{Z}}
\newcommand{\MFACVAR}{\FL}
\newcommand{\MSVAR}{\LAS}

\newcommand{\MFACABCAR}{\FBRAR}
\newcommand{\MAR}{\ARBIC}

\newcommand{\THRarg}[1]{\operatorname{THR}_{#1}}

\newcommand{\ARBIC}{\emph{$\boldsymbol{AR_{BIC}}$}}
\newcommand{\FBNfilt}{\emph{$\boldsymbol{F_{BN}^{ARfilt}}$}}
\newcommand{\FLfilt}{\FL}
\newcommand{\FL}{\emph{$\boldsymbol{FL_{sel}}$}}
\newcommand{\FBRAR}{\emph{$\boldsymbol{F_{BN}AR}$}}
\newcommand{\LAS}{\emph{$\boldsymbol{L_{sel}}$}}

\theoremstyle{definition}

\newtheorem{assumption}{Assumption}
\newtheorem{example}{Example}
\newtheorem{remark}{Remark}
\newtheorem{theorem}{Theorem}
\newtheorem{lemma}{Lemma}

\theoremstyle{plain}
\makeatletter
\def\adl@drawiv#1#2#3{%
        \hskip.5\tabcolsep
        \xleaders#3{#2.5\@tempdimb #1{1}#2.5\@tempdimb}%
                #2\z@ plus1fil minus1fil\relax
        \hskip.5\tabcolsep}
\newcommand{\cdashlinelr}[1]{%
  \noalign{\vskip\aboverulesep
           \global\let\@dashdrawstore\adl@draw
           \global\let\adl@draw\adl@drawiv}
  \cdashline{#1}
  \noalign{\global\let\adl@draw\@dashdrawstore
           \vskip\belowrulesep}}
\makeatother



\begin{document}
\bibliographystyle{apalike2}

\begin{frontmatter}
\title{\textcolor{black}{Factor Models with Sparse VAR Idiosyncratic Components}}
\author{Krampe, J.$^{\dagger}$,\; Margaritella, L.$^{\ddagger}$}
\address{$^{\dagger}$University of Mannheim,\; $^{\ddagger}$Lund University\\
\href{mailto:krampe@uni-mannheim.de}{j.krampe@uni-mannheim.de}, \href{mailto:luca.margaritella@nek.lu.se}{luca.margaritella@nek.lu.se}
}
\date
\maketitle
\begin{abstract}
We reconcile the two worlds of dense and sparse modeling by exploiting the positive aspects of both. We employ a factor model and assume {the dynamic of the factors is non-pervasive while} the idiosyncratic term follows a sparse vector autoregressive model (VAR) {which allows} for cross-sectional and time dependence. The estimation is articulated in two steps: first, the factors and their loadings are estimated via principal component analysis and second, the sparse VAR is estimated by regularized regression on the estimated idiosyncratic components. We prove the consistency of the proposed estimation approach as the time and cross-sectional dimension diverge. In the second step, the estimation error of the first step needs to be accounted for. Here, we do not follow the naive approach of simply plugging in the standard rates derived for the factor estimation. Instead, we derive a more refined expression of the error. This enables us to derive tighter rates. We discuss the implications of our model for forecasting, factor augmented regression, bootstrap of factor models, and time series dependence networks via semi-parametric estimation of the inverse of the spectral density matrix.
\bigbreak
\end{abstract}
\end{frontmatter}

\section{Introduction} 
In the past twenty years, factor models have emerged as a major tool for the analysis and forecast of high-dimensional time series. Such models are characterized by the assumed existence of a specific decomposition of the high-dimensional vector of time series into two mutually orthogonal unobserved components. The \textit{common component}, driven by a finite number of (possibly dynamic) factors, 
represents the comovements among the series and it is of reduced rank. The weakly cross-correlated  \textit{idiosyncratic component} (cf.~\textit{generalized} factor model) represents individual features of the series.
\footnote{Conversely to the \textit{generalized} factor model, the \textit{exact} factor model assumes no cross-sectional dependence in the idiosyncratics, thus working with the assumption of a diagonal cross-autocovariance matrix.} The ramification of the literature on factor models is mostly due to the way this decomposition is characterized. Traditionally, the common component is identified as those eigenvalues of the covariance matrix which diverge while the idiosyncratics are those which stay bounded. \citet{forni2000generalized} and \citet{forni2015dynamic} assume the explosive eigenvalues to be also dynamic, reflecting both a contemporaneous and lagged effect of the common components on the series. \citet{stock2002forecasting,bai2002determining,bai2019rank} assume explosive static eigenvalues (i.e.~only contemporaneous effects of the factors) along with a finite-dimensional factor space. 
One thing is certain: the two components are radically different objects that need to be treated differently to better capture their respective dependence structures. Assuming away any cross-sectional dependence of the idiosyncratic components can be misleading. Some consider them as simple univariate autoregressive processes \citep[see e.g.,][]{forni2005generalized}  or even as white noise processes and drop them when producing forecasts \citep[see e.g.,][]{lam2012factor}. This however neglects the predictive power that idiosyncratic components can have, thus resulting in less accurate forecasts.
\par In this paper we lean forward towards a reduced rank plus sparse characterization of the factor model decomposition by assuming the idiosyncratic component to follow a high-dimensional vector autoregressive (VAR) model. This allows cross-sectional and time dependence in the idiosyncratic term. In the first estimation step, we employ principal component analysis (PCA) to estimate the factors. In a second step, high-dimensional penalized VAR through the (adaptive) lasso is used in order to estimate the idiosyncratic components. By so doing we combine a ``dense" modeling approach for the factors with a ``sparse" modeling approach for the idiosyncratic component.\footnote{The terms \textit{dense} and \textit{sparse} are used to distinguish estimation approaches that require or not the structural assumption of a sparse coefficient vector to perform some dimensionality reduction \citep[see also][]{giannone2017economic}. PCA is therefore a dense approach while lasso is a sparse one.} Thus, allowing for a more refined disentangling of the dependence structure of the two components. We show the consistent estimation of both the sparse VAR model driving the idiosyncratic components and the factors, as both the cross-sectional and time dimensions grow large. When estimating the sparse VAR for the idiosyncratic component, a naive approach would be to simply plug in the standard rates derived for the factor estimation. This however leads to a suboptimal rate. Instead, an important contribution of our work is deriving detailed expressions of the occurring errors. This enables us to obtain tighter rates for the second step and also employ a semi-parametric estimator for the inverse of the spectral density matrix. 
We discuss the implications of our proposed framework for forecasting, factor-augmented regression, bootstrap of factor models, and the estimation of time series dependence networks. We also propose a joint information criterion that combines the approach of \citet{bai2002determining} with an extra penalty allowing for simultaneous lag-length estimation of the VAR model. The benefit of our proposed procedure is confirmed through extensive simulations where different levels of sparsity, number of factors, lag-length of the VARs, and idiosyncratic covariance matrix are considered for different sample sizes and dimensions. We also compare our combined procedure with the standard high-dimensional forecasting methods which fully rely on either a sparse or a dense procedure.
\par \par There already exists applications in the literature combining (dynamic) factor models with sparse vector autoregressive models, see e.g., \cite{barigozzi2017network,barigozzi2019nets} and more recently \citet{barigozzi2022fnets}. However, \cite{barigozzi2017network,barigozzi2019nets} do not present theoretical results about the combined approach and the framework considered in \cite{barigozzi2022fnets} differs in important aspects from the one considered here, see Section~\ref{sec_model} and the discussion after Assumption~\ref{sparsity} for details. In the non-dynamic idiosyncratics set-up, \cite{kneip2011factor,fan2020factor} combine factors with regularized models. Since regularized methods such as the lasso have difficulties with strongly correlated regressors, especially in the context of model selection, they aim to decorrelate the regressors by adjusting for the factors. 
Furthermore, \cite{fan2021bridging} provide hypothesis tests to check whether after removing factors (as well as trends in a first step) the regressors possess some pre-defined weakly correlated structure or not. 
\cite{fan2020factor,fan2021bridging} allow for time-dependent regressors, however, they do not consider nor allow that the idiosyncratic part follows a sparse vector autoregressive model where the cross-sectional sparsity can grow with the sample size.
In the context of high-dimensional VAR models,
another approach is to consider the slope matrices as a combination of a low-rank matrix and a sparse matrix as done in \cite{basu2019low}. The low-rank part takes here a similar role as the common component of the factor model and it is estimated by nuclear-norm regularization. In the context of high-dimensional VAR models with strong cross-sectional correlated noise, a combination of low-rank plus sparse has been also explored by \citet{lin2019approximate} and more recently in \citet{miao2022high}. However, these approaches differ from ours in terms of model and estimation approach. A more detailed discussion can be found in Remark~\ref{subsection.low-rank} in Section~\ref{section.estimation}.
\par The remainder of the paper is organized as follows: Section \ref{sec_model} introduces the factor model with sparse VAR idiosyncratic components and reports few standard assumptions defining its behavior. Section \ref{section.estimation} is devoted to describing the two-step procedure used to estimate the factor model with sparse VAR idiosyncratic components and prove its consistency. Theorem \ref{thm.est.error.step.one} derives a representation of the idiosyncratic components estimation error while Theorem \ref{thm.est.error.idio} is the main result establishing bounds for the estimation error for the second step of the estimation procedure i.e., for the lasso on the sample estimates of the idiosyncratic component. The same two-step procedure with mild additional assumptions can also be employed in estimating the spectral density of the process and Theorem \ref{thm.spectral.density} derives the relative estimation error bounds. { Section \ref{sec_applications} discusses the implications of our model for forecasting, factor augmented regression, bootstrapping factor models, and the estimation of time series dependence networks. The latter point is illustrated by estimating networks based on the partial coherence for the FRED-MD data set.}
Section \ref{sec_numbfactors} considers the problems of: estimating the number of factors, determining the lag-length in the VAR, and tuning the penalty parameter for the lasso. Section \ref{Sec_Simulations} reports simulation results for our proposed method under different VAR data generating processes in terms of design and sparsity.
Finally, Section \ref{sec_conclusions} concludes.
\bigbreak
A few words on notation. Throughout the paper we use boldface characters to indicate vectors and boldface capital characters for matrices. For any $n$-dimensional vector $\boldsymbol{x}$, we let $\norm{\boldsymbol{x}}_p = \left(\sum_{i=1}^n |x_i|^p \right)^{1/p}$ denote the $\ell_p$-norm and $\be_j=(0,\ldots,0,1,0,\ldots, 0)^\top$ denotes a unit vector of appropriate dimension with the one appearing in the $j$th position. Furthermore, for a $r\times s$ matrix $\BS A=(a_{i,j})_{i=1,\ldots,r, j=1,\ldots,s}$,  $\|\BS A\|_1=\max_{1\leq j\leq s}\sum_{i=1}^r|a_{i,j}|=\max_j \| \BS A \be_j\|_1$, $\|\bA\|_\infty=\max_{1\leq i\leq r}\sum_{j=1}^s|a_{i,j}|=\max_{i} \| \be_i^\top \BS A\|_1$ and $\|\BS A\|_{\max}=\max_{i,j} |\be_i^\top \BS A \be_j|$.
$\bA^i$ denotes the $i$th matrix power of $\bA$ and $\bA^{(i)}$ refers to the $i$th element of a sequence of matrices. We denote the largest absolute eigenvalue of a square matrix $\BS A$ by $\sigma_{\max}(\BS A)$ and let   $\|\BS A \|_2^2=\sigma_{\max}(\BS A \BS A^\top)$. We denote the smallest eigenvalue of a matrix $\bA$ by $\sigma_{\min}(A)$.
For any index set $S \subseteq \{1, \ldots, n\}$, let $\boldsymbol{x}_{S}$ denote the sub-vector of $\boldsymbol{x}_t$ containing only those elements $x_i$ such that $i \in S$. $\|\BS x\|_0$ denotes the number of non-zero elements of $\BS x$.

\section{The Model}\label{sec_model}
We work {with a generalized factor model} 
where both factors and idiosyncratic components are allowed to be {(second order)} stationary stochastic processes and the loadings are static.
{To elaborate,}
let $\boldsymbol{x}_t=(x_{1,t},\ldots,x_{N,t})^\top$, $t=1,\ldots, T$, be a $N\times T$ rectangular data array  representing a finite realization of an underlying real-valued stochastic process $\{x_{i,t}\}.$ 
Assume that for each $t$, $\boldsymbol{x}_t$ can be decomposed into a sum of a common component $\boldsymbol{\chi}_t=(\chi_{1,t}, \dots, \chi_{N,t})^{\top}$ and an idiosyncratic component $\boldsymbol{\xi}_t=(\xi_{1,t}, \dots, \xi_{N,t})^{\top}$, both latent {and mutually orthogonal at all leads and lags}. Then, the factor model decomposition takes the following usual form\small
\begin{equation}
    \boldsymbol{x}_t=\boldsymbol{\chi}_t+\boldsymbol{\xi}_t. \label{eq.fac.decom}
\end{equation}\normalsize
{The common component has reduced rank i.e.,} $\BS \chi_{t}$ is driven linearly by an $r$-dimensional vector of common factors $\boldsymbol{f}_t=(f_{1,t},\ldots,f_{r,t})^{\top}$, where $r$ is considered as fixed as both the cross sectional dimension $N$ and the time series dimension $T$ grow large and $r\ll N$. 
The common components $\chi_{i,t}$ can {then be represented by} the following linear combination\small
\begin{equation}\label{hdm}
    \chi_{i,t}=\sum_{s=1}^r \ell_{i,s}f_{s,t}=\boldsymbol{\Lambda}^{\top}_i\boldsymbol{f}_t,
\end{equation}\normalsize
where $\ell_{i,s}$ are denoted as loadings. Note that $\chi_{i,t}$ is uniquely defined. But since for {any} rotation matrix $\bH$,   $\chi_{i,t}=\boldsymbol{\Lambda}^{\top}_i \bH \bH^{-1}\boldsymbol{f}_t$ is a valid linear combination as well, $\boldsymbol{\Lambda}^{\top}_i, \boldsymbol{f}_t$ are only identified up to some arbitrary rotation.

We do not assume that the factors $\boldsymbol{f}_t$ nor the idiosyncratic component $\boldsymbol{\xi}_t$ are independent and identically distributed but we allow them to be stationary stochastic processes. We consider that the factors are given by a one-sided linear process, see Assumption~\ref{ass.moments} below. This includes the cases that the factors are driven by a stable vector autoregressive (VAR) model. Additionally, the idiosyncratic component $\boldsymbol{\xi}_t$ {is allowed to be weakly cross-correlated and we consider this} to follow a sparse  VAR model of order $p$ as \small
\begin{equation}\label{idio}
    \boldsymbol{\xi}_t=\sum_{j=1}^p \bA^{(j)}\boldsymbol{\xi}_{t-j}+\boldsymbol{v}_t=\sum_{j=0}^{\infty}\bB^{(j)}\boldsymbol{v}_{t-j},
\end{equation}\normalsize
for $\bv_t$ being a white noise process, $\bA^{(j)}$ the sparse slope matrices and $\bB^{(j)}$ the moving average matrices of the vector moving average (VMA)$(\infty)$-representation of the VAR$(p)$ model; see Assumption \ref{sparsity} and \ref{ass.moments} below for details on sparsity and moment conditions. This includes as special case that the idiosyncratic components are driven by individual univariate autoregressive processes or even i.i.d..

{Let us note that even if the factors are dynamic, the relationship between $\boldsymbol{x}_t$ and $\boldsymbol{f}_t$ is assumed here to be static. This type of factor model is widely used in practice and it differs from the framework of \citet{forni2000generalized} which assumes a pervasive dynamics of the common factors where $\boldsymbol{x}_t$ is set to also depend on $\boldsymbol{f}_t$ with lags in time. In several cases, though it is possible to transform a dynamic relationship into a stacked static relationship, see among others Section 2.1.2 in \cite{stock2016dynamic}. Especially, if the assumption of a finite-dimensional span of the common component is used, then one can cast a dynamic representation into a static one \citep[see e.g.,][]{bai2007determining}. Alternatively, block-VAR filtering of $\bx_t$ as proposed in \citet{forni2015dynamic} also allows to turn lagged loadings into static ones.}

In the following Assumptions~\ref{sparsity}, \ref{ass.moments}, and \ref{ass.fac}, the sparsity and stability conditions, the factors, moment conditions, and loadings are further specified. 
\bigbreak
\begin{assumption}\label{sparsity}(\textit{Sparsity and stability})\\
$(i)$ Let $\A$  denote the stacked (companion) VAR matrix of (\ref{idio}). Let $k$ denote the row-wise sparsity of $\A$ with approximate sparsity parameter $q \in [0,1)$, i.e.,\footnote{{$q=0$ corresponds to the usual \textit{exact sparsity} assumption where several parameters are exactly zero. \textit{Approximate sparsity} $q>0$ allows for many parameters not to be exactly zero but rather small in magnitude.}} \small
$$\max_i \sum_{s=1}^p \sum_{j=1}^N |\BS A_{i,j}^{(s)}|^q=\max_i \sum_{j=1}^{Np} |\A_{i,j} |^q\leq k.$$\normalsize
$(ii)$ The VAR process is considered as stable such that for a constant $\rho \in (0,1)$ we have independently of the sample size $T$ and dimension $N$: $\|\A^j\|_{2}=\sqrt{\sigma_{\max}(\A^{j\top}\A^{j})}\leq M \rho^j$, 
where $M$ is some finite constant. 
Additionally, we have $\|\BS \Gamma_\xi(0)\|_\infty\leq k_\xi M$, where $\BS \Gamma_\xi(0)=\var(\BS \xi_t)$ and $\sigma_{\min}(\var((\BS \xi_t^\top,\dots,\BS\xi_{t-p+1}^\top)^\top))>\alpha>0$. The sparsity parameter $k$ as well as $k_\xi$ are allowed to grow with the sample size. 
\end{assumption}

Note that Assumption \ref{sparsity},(i) is quite general as the sparsity is \emph{row-wise} and {is allowed} to grow with the sample size. Let us emphasize that the weaker assumption of approximate sparsity instead of exact sparsity (i.e., $q=0$), is used throughout. In the context of forecasting, the assumption of a stable and \emph{row-wise sparse} VAR model is standard in the literature of sparse VAR models, see among others \cite{kock2015oracle,han2015direct,masini2019regularized}. When the focus is on estimating the dependency structure, e.g., spectral density matrices, additional \emph{column-wise sparsity} seems unavoidable, see \cite{krampe2020statistical} for a discussion of different sparsity concepts for VAR models. Here, we only require additional \emph{column-wise sparsity}
 in Section \ref{sec_spec_dens}, where spectral density estimation is discussed. Row-wise sparsity (with or without additional column sparsity) includes as special case the univariate autoregressive model for each idiosyncratic component. The latter is generally not allowed, if the sparsity condition is specified for the entire matrix, e.g., $\sum_{s=1}^p \sum_{i,j=1}^N |\BS A_{i,j}^{(s)}|^q \leq k$, see Section~2 in \cite{krampe2020statistical} for further discussion.
 
{The common and idiosyncratic components are identified based on the diverging behavior of the eigenvalues of the covariance matrix. This implies restrictions with respect to  the matrix-norm $\|\cdot\|_2$ but not with respect to $\|\cdot\|_\infty$. Hence, as $\bxi$ is an idiosyncratic component, $\|\BS \Gamma_\xi(0)\|_2$ is bounded but $\|\BS \Gamma_\xi(0)\|_\infty$ can still grow with dimension. Consequently, when modeling the idiosyncratic component with a VAR model such behavior should not be ruled out by over-restricting the VAR slope matrices. Since row-wise sparsity of the slope matrices allows that $\|\A\|_\infty$ can grow with sparsity parameter $k$, Assumption~\ref{sparsity} allows growth in $\|\BS \Gamma_\xi(0)\|_\infty$ as it is specified by the (possibly growing) parameter $k_\xi$. }

{If one would work within the framework of \citet{forni2017dynamic}, their assumptions on the serial and cross-sectional dependence of the idiosyncratic terms are more restrictive. To be specific, Assumption~4 in \cite{forni2017dynamic} impose $\|\cdot\|_1$- and $\|\cdot\|_\infty$-boundedness on the idiosyncratic coefficient matrices. This would imply that $k_\xi$ is bounded and
in our case would also imply the slope VAR matrices to be bounded in matrix-norm $\|\cdot\|_\infty$. Hence, this would restrict the sparsity of the VAR slope matrices to be (more or less) fixed which would be less general and not desirable. Let us mention that this framework of dynamic factor models with boundedness conditions on the idiosyncratic coefficient matrices  is considered in \cite{barigozzi2022fnets}. An extension of the \citet{forni2017dynamic} assumption to allow for growing sparsity is beyond the scope of this paper and is left for future research.}

In the established literature on factor estimation, a common assumption is to restrict the growth of the linear dependence of the idiosyncratic component. For instance, Assumption~3C in \cite{bai2003} states that the absolute sum of all covariances of the idiosyncratic component grow with order $N$. Here, we quantify the linear dependence of the idiosyncratic component by the condition $\|\Gamma_\xi(0)\|_\infty=\max_i \sum_{j=1}^N |\cov(x_{i,t},x_{j,t})|\leq k_\xi$ and the object $k_\xi$. Since $\|\Gamma_\xi(0)\|_\infty\leq \sqrt{N} \|\Gamma_\xi(0)\|_2$, $\sqrt{N}$ is an upper bound for the growth rate of $k_\xi$. 
As discussed previously, a bounded $k_\xi$ can be too restrictive hence we consider that $k_\xi$ grows moderately. We do not specify here a rate for $k_\xi$ but a rate smaller than $\sqrt{N}$ seems most realistic and would be more in line with established assumptions in the factor literature. The reason for this is that a growth rate of $\sqrt{N}$ would allow the absolute sum of all covariances could grow with a rate $N^{3/2}$. This would violate, for instance, the previously mentioned  Assumption~3C in \cite{bai2003}. 
Note further that if the maximal absolute row sum of the covariance of the idiosyncratic component is growing way faster than the average row sum, we may end up in the context of weak factors, see among others \cite{onatski2012asymptotics}.
Note that using the VMA$(\infty)$-representation, Assumption~\ref{sparsity} gives  the upper bound $k^2 \|\BS \Sigma_v\|_\infty$ for $k_\xi$, where $\BS \Sigma_v=\var(\BS v_t)$ is the variance matrix of the residuals of the idiosyncratic component. 

\begin{assumption}\label{ass.moments} (\textit{Factor dynamics and moments})\\
The factors are given by a one-sided linear filter with geometrically decaying coefficients, that is:\small \begin{equation*} \boldsymbol{f}_t=\sum_{j=0}^\infty \BS D^{(j)} \bu_{t-j}, \end{equation*}\normalsize and $\|\BS D^{(j)}\|_2\leq K \rho^j$, where $K$ is some positive constant and $\rho\in(0,1)$. 
Furthermore, $\{(\bu_t^\top,\bv_t^\top)^\top, t \in \Z\}$ is an i.i.d.~sequence and $\cov(\bu_t,\bv_t)=0$. 
Let $\zeta> 8$ be the number of finite moments of $\{(\bu_t^\top,\bv_t^\top)^\top, t \in \Z\}$, i.e.,  $\mathbb{E} |\bu_{t,j}|^\zeta\leq M$ and
$ \max_{\|\bw\|_2\leq 1} \mathbb{E}  |\bw^\top \boldsymbol{v}_t|^\zeta\leq M$. We denote $\bSigma_u=:\var(\bu_t)$ and $\bSigma_v=:\var(\bv_t)$.
\end{assumption}

\begin{assumption} \label{ass.fac}(\textit{Factors and loadings})\\
Let $M$ be some finite constant, then
\begin{enumerate}
    \item $\lim_{T\to \infty} 1/T\sum_{t=1}^T \boldsymbol{f}_t \boldsymbol{f}_t^\top=\mathbb{E}[ \boldsymbol{f}_t \boldsymbol{f}_t^\top]=\bSigma_F \in \R^{r\times r}$ positive definite  and $\|\bSigma_F\|_2\leq M$. 
    \item $\lim_{N \to \infty} 1/N \sum_{i=1}^N \bLambda_i \bLambda_i^\top=\bSigma_\Lambda \in \R^{r\times r}$, positive definite with largest eigenvalue $\sigma_{\Lambda,\max}\leq M$ and smallest eigenvalue $\sigma_{\Lambda,\min}\geq 1/M>0$ , $\|1/N \sum_{i=1}^N \bLambda_i \bLambda_i^\top\|_2\leq M$ for all $N$.
    \item All eigenvalues of $\bSigma_F,\bSigma_\Lambda$ are distinct.
\end{enumerate}
\end{assumption}

Assumption~\ref{sparsity} and \ref{ass.moments} imply that $\{\bxi_t\}$ is stationary and let the autocovariance function be given by $\boldsymbol{\Gamma}_{\xi}(s-t)=\cov(\boldsymbol{\xi}_{s},\boldsymbol{\xi}_t)$. Furthermore, Assumption~\ref{ass.moments} implies that the factors are also a stationary process such that $\{\boldsymbol{x}_t\}$ itself is indeed stationary. In order to quantify the dependence of stochastic processes, we use the concept of functional dependence, see \cite{wu2005nonlinear}. Since this is only necessary for the proofs, we do not introduce the notation here and refer to Remark~\ref{remark.functional.dependence} in the appendix. 


The moment condition in Assumption \ref{ass.moments} refers to the situation in which only a finite number of moments, here $\zeta$, are finite. Hence, we do not assume sub-Gaussian processes or similar, which is often assumed for sparse VAR processes, see \cite{basu2015,kock2015oracle,han2015direct}. For sub-Gaussian processes, the polynomial terms depending on $\zeta$ would vanish, which would result in tighter error bounds obtained later on. The reason for only assuming $\zeta$ finite moments is to be more in line with the classical factor literature, see among others \cite{bai2003,stock2002forecasting,forni2000generalized,forni2017dynamic}. E.g., \cite{bai2003} derived inferential results for factor models under $8$th finite moments of the idiosyncratic part and $4$th finite moments of the factors. Note that the filter in Assumption~\ref{ass.moments} can be the one-sided representation of a stable VARMA model as in Assumption~2 in \cite{forni2017dynamic}. 
Assumption~\ref{ass.fac} is a standard assumption in the context of strong factor models, see \cite{stock2002forecasting,bai2003}. It implies that each of the factors provides a non-negligible contribution to the variance of each component of $\{\BS x_t\}$. We would like to point out here that the time and cross-sectional dependence of the idiosyncratic component is only limited by assuming that it follows a sparse VAR model. Furthermore, it is not clear if assuming a sparse VAR model for the idiosyncratic part is a special case of the assumptions to time and cross-section dependence in the factor literature, see among others Assumption~C in \cite{bai2003}. 
The reason for this is that the sparsity is not fixed but it can grow with the sample size. Nevertheless, the error bounds obtained later on requires that the sparsity cannot grow too fast with increasing dimension.

\section{Estimation} \label{section.estimation}
In this section we propose a two-step approach to estimate a factor model  with sparse VAR idiosyncratic components and prove its consistency. For this, let  $\boldsymbol{x}_t,$ $t=1,\dots,T$ be some observations and let $\boldsymbol{X}=\boldsymbol{\chi}+\boldsymbol{\Xi}$ denote the $T\times N$ matrix form of \eqref{eq.fac.decom}. Furthermore, $\boldsymbol{\Lambda}$ denotes the $N\times r$ matrix of loadings and $\boldsymbol{F}$ denotes the $T \times r$ matrix of factors such that $\boldsymbol{\chi}=\boldsymbol{F} \boldsymbol{\Lambda}^\top$ is the matrix counterpart of \eqref{hdm}.
Then, an estimation of the factor decomposition can be obtained by using Principal Components Analysis (PCA), see among others \cite{bai2003,bai2020simpler}. The number of factors $r$ is considered as known here. {Note that the number of factors can be determined by various approaches (see Section~\ref{sec_numbfactors} for further discussion) and it can be estimated with probability tending to one, see among others  \cite{bai2002determining}.}
To elaborate with the estimation, let \small$$\boldsymbol{X}/\sqrt{NT}=\BS U_{NT} \BS D_{NT} \BS V_{NT}^\top,$$\normalsize denote a singular value decomposition of $\boldsymbol{X}/\sqrt{NT}$ such that $\BS D_{NT}$ is a diagonal matrix with the singular values arranged in descending order on its diagonal. $\BS U_{NT}$ and $\BS V_{NT}$ are the corresponding left and right singular vectors, respectively. This can be further written as \small$$\BS U_{NT} \BS D_{NT} \BS V_{NT}^\top=\BS U_{NT,r} \BS D_{NT,r} \BS V_{NT,r}^\top+\BS U_{NT,N-r} \BS D_{NT,N-r} \BS V_{NT,N-r}^\top,$$\normalsize where $\BS D_{NT,r}$ is a diagonal matrix with the first $r$ largest singular values, $d_{NT,1},\dots,d_{NT,r}$, arranged in descending order on its diagonal, $\BS D_{NT,N-r}$ is a diagonal matrix with the remaining $N-r$ largest singular values, and $\BS U_{NT,r},\BS U_{NT,N-r},\BS V_{NT,r},\BS V_{NT,N-r}$ are the corresponding left and right singular vectors. Then,  the estimators of a rotated version of $\boldsymbol{F}$ and $\boldsymbol{\Lambda}$ are given by \small
$$\hat{\boldsymbol{F}}=\sqrt{T} \bU_{NT,r} \text{ and } \hat{\boldsymbol{\Lambda}}=\sqrt{N} \BS V_{NT,r} \BS D_{NT,r},$$\normalsize such that  $\hat{\boldsymbol{\chi}}=\hat{\boldsymbol{F} }\hat{\boldsymbol{\Lambda}}^\top$ and $\hat {\BS \xi}=\BS x_t-\hat{\boldsymbol{\chi}} $. This uses the normalization $\hat {\boldsymbol{F}}^\top \hat{\boldsymbol{F}}/T=\bI_r$ and $\hat{\boldsymbol{\Lambda}}^\top \hat{\boldsymbol{\Lambda}}$ is a diagonal matrix. Consider the estimated idiosyncratic components  $\hat \bxi$. As it is assumed that $\{\BS \xi_t\}$ follows a sparse vector autoregressive model, we estimate this sparse VAR on $\hat \bxi$ by regularized methods such as the (adaptive) lasso. This idea leads to the following two-step estimation procedure:

\begin{enumerate}
    \item Perform a singular value decomposition of\small $$\boldsymbol{X}/\sqrt{NT}=\BS U_{NT,r} \BS D_{NT,r} \BS V_{NT,r}^\top+\BS U_{NT,N-r} \BS D_{NT,N-r} \BS V_{NT,N-r}^\top,$$ \normalsize where $\BS U_{NT,r} \BS D_{NT,r} \BS V_{NT,r}^\top$ corresponds to the first $r$ singular values. \\
    Set $\hat{\boldsymbol{F}}=\sqrt{T} \bU_{NT,r}$, $\hat{\boldsymbol{\Lambda}}=\sqrt{N} \BS V_{NT,r} \BS D_{NT,r}$, and $\hat {\BS \xi}=\BS x_t-\hat{\boldsymbol{F} }\hat{\boldsymbol{\Lambda}}^\top $. 
    \item Let $\hat {\BS{\xi}}_t^v=(\hat {\BS\xi}_t^\top,\dots,\hat {\BS\xi}_{t-p}^\top)^\top$. Then, an adaptive lasso estimator for $\bbeta^{(j)}$ i.e., the $j$th row of\\ $(\bA^{(1)},\dots,\bA^{(p)}),$ is given by \small
\begin{align}
    \label{eq.lasso.beta.j}
    \hat{\bbeta}^{(j)}=\argmin_{\bbeta \in \R^{Np}} \frac{1}{T-p} \sum_{t=p+1}^T (\hat \xi_{j,t}-\bbeta^\top \hat {\BS\xi}_{t-1}^v)^2+\lambda \sum_{i=1}^N |g_i \beta_i|,\quad j=1,\dots,N, 
\end{align}\normalsize
            where $\lambda$ is a non-negative tuning parameter which determines the strength of the penalty and $g_i, i=1,\dots,N,$ are weights. For instance, $g_i=1$ leads to the standard lasso. Let also $(\hat \bA^{(1)},\dots,\hat \bA^{(p)})$ be the matrices that correspond to stacking $\hat {\bbeta}^{(j)},j=1,\dots,N$.
\end{enumerate}

By estimating the factors $\boldsymbol{f}_t$ through standard Principal Components Analysis (PCA) and the sparse VAR models of the idiosyncratic components $\boldsymbol{\xi}_{t}$ via sparse penalized regression techniques, we combine a dense estimation approach with a sparse one. This can possibly better capture and disentangle both the dependence coming from the diverging eigenvalues of $\mathbb{E}(\boldsymbol{X}\boldsymbol{X}^{\top})$, i.e., the factors, as well as the dependence coming from the non-diverging eigenvalues of $\mathbb{E}(\boldsymbol{X}\boldsymbol{X}^{\top})$, i.e., the idiosyncratic components.
The estimation of factors and loadings via PCA is a well established method in the literature, see among others \cite{stock2002forecasting,bai2003}, and the common and idiosyncratic component can be estimated with rate $O_P(\max(1/\sqrt{T},1/\sqrt{N}))$. Since it is no different in this setting, we focus our presentation on the second estimation step. For a sparse stationary VAR model, deviation bounds and restricted eigenvalue conditions can be established, see among others \cite{basu2015,kock2015oracle}. Given these, the consistency of the lasso can be derived and, under additional Gaussianity assumption, one obtains a rate for strict sparsity $k$ of $O_P(k \sqrt{\log(N)/T})$. However, as the idiosyncratic component $\{\BS \xi_t\}$ is not observed in our setting and hence needs to be estimated, the regression in Step 2 is performed only with the estimated idiosyncratic component. Consequently, the aforementioned results cannot be applied here. Before analyzing the second step, we in fact need to quantify the estimation error $\BS w_t:=\BS {\hat \xi_t}-\BS \xi_t\equiv \BS F \BS \Lambda^{\top}-\hat{\BS F} \hat{\BS \Lambda}^{\top}$ arising from the first step. For the consistency of the lasso, this means quantifying the estimation error $\BS w_t$ in quantities such as $\|1/T \sum_{t=1}^T (\BS \xi_t+\BS w_t)(\BS \xi_t+\BS w_t)^\top\|_{\max}$. If we simply apply the rate derived in the literature for approximate factor models, see among others \cite{stock2002forecasting,bai2003} which derive $\BS w_t=O_P(\max(1/\sqrt{T},1/\sqrt{N}))$, we  would obtain $\|1/T \sum_{t=1}^T (\BS \xi_t+\BS w_t)(\BS \xi_t+\BS w_t)^\top\|_{\max}=\|1/T \sum_{t=1}^T \BS \xi_t\BS \xi_t^\top\|_{\max}+O_P(\max(1/\sqrt{T},1/\sqrt{N}))$. This may lead to a rate for the second step of $O_P(k (\sqrt{\log(N)/T}+k\max(1/\sqrt{N},1/\sqrt{T}))$. However, this can be improved if we analyze the estimation error $\BS w_t$ more closely. For this, we follow the idea of the decomposition in eq. (6) in \cite{bai2020simpler}.
To elaborate, we have $1/(NT) \BS X \BS X^\top \hat{\BS F}=\hat{\BS F} \BS D_{NT,r}^2$. Plugging in \eqref{eq.fac.decom}, and using the rotation matrix \small
\begin{align}
\BS H_{NT}^\top=(\boldsymbol{\Lambda}^\top \boldsymbol{\Lambda}/N) (\boldsymbol{F}^\top \hat{\boldsymbol{F}}/T) \BS D_{NT,r}^{-2}, \label{eq.rotation.matrix}    
\end{align}\normalsize
we obtain the following representation for the error between the estimated factors and a rotated version of the factors
\small
\begin{align}
\hat{\boldsymbol{f_t}}-\BS H_{NT}\boldsymbol{f_t}=&\frac{1}{NT}\Bigg[ \sum_{i=1}^N \sum_{s=1}^T \BS f_t^\top \BS\Lambda_i \xi_{i,s} \hat{\BS f_s}+ \sum_{i=1}^N \sum_{s=1}^T \xi_{i,t} \BS\Lambda_i \BS f_{s}^\top \hat{\BS f_s}+ \sum_{i=1}^N \sum_{s=1}^T \xi_{i,t}\xi_{i,s}  \hat{\BS f_s}\Bigg] \BS D_{NT,r}^{-2}
\label{eq.factor.diff}.
\end{align}\normalsize
Similarly, we obtain by symmetry for the loadings\small
\begin{align*}
        (\BS H_{NT}^\top)^{-1} \BS \Lambda_i - \hat{\BS \Lambda}_i=&\frac{1}{T} \Bigg[\sum_{s=1}^T \BS  H_{NT} \BS f_s \xi_{i,s}+ \sum_{s=1}^T (\hat{\BS f}_s - \BS H_{NT} {\BS f}_s) \xi_{i,s}+\sum_{s=1}^T \BS H_{NT}\BS f_s [\hat{\BS f}_s - \BS H_{NT} \BS f_s]^\top (\BS H_{NT}^T)^{-1} \BS \Lambda_i \\
    & \sum_{s=1}^T [\hat {\BS f}_s - \BS  H_{NT} {\BS f}_s] [\hat{\BS f}_s - \BS H_{NT} \BS f_s]^\top (\BS H_{NT}^T)^{-1} \BS \Lambda_i \Bigg]. 
\end{align*}\normalsize
These representations can be used to derive the order of the estimation error for the factors and loadings as it is done with a slightly different rotation matrix in \cite{bai2003}. However, as our focus is on $\BS w_t:=\BS {\hat \xi_t}-\BS \xi_t$, we use these results to derive a representation of $\BS w_t$. With the obtained representation for $\BS w_t$, we can analyze more closely the estimation error of the second step. For this, note first that $\|1/T \sum_{t=1}^T (\BS \xi_t+\bw_t)(\BS \xi_t+\bw_t)^\top\|_{\max}\leq \|1/T \sum_{t=1}^T (\BS \xi_t)(\BS \xi_t)^\top\|_{\max}+2\|1/T \sum_{t=1}^T (\bw_t)(\BS \xi_t+\bw_t)^\top\|_{\max}+\|1/T \sum_{t=1}^T (\bw_t)(\bw_t)^\top\|_{\max}$. Bounds for the latter objects and the representation of $\BS w_t$ are given in the following Theorem~\ref{thm.est.error.step.one}.

\begin{theorem} \label{thm.est.error.step.one}
Under Assumption \ref{sparsity}, \ref{ass.moments}, and \ref{ass.fac}, we have for $t=1,\dots,T,j=1,\dots,N$
and\small
\begin{align}
    w_{j,t}:=\hat{ \xi}_{j,t}-\xi_{i,t}=& \BS \Lambda_j^\top \BS H_{NT}^{-1} \frac{1}{NT}\left[\sum_{i=1}^N \sum_{s=1}^T \xi_{i,t} \BS\Lambda_i \BS f_{s}^\top \BS H_{NT}{\BS f_s}+ \sum_{i=1}^N \sum_{s=1}^T \xi_{i,t}\xi_{i,s} \BS H_{NT} {\BS f_s}\right] \BS D_{NT,r}^{-2} \nonumber\\
    &+ \BS f_t^\top \BS H_{NT}^\top  \frac{1}{T} \left[\sum_{s=1}^T \bH_{NT} \BS f_s \xi_{j,s}\right]+Error_j, \label{corrupt_repres}
    \end{align}\normalsize
where \small
$$\max_j |Error_j|=O_P\left(\frac{{\log(N)}}{{T}}+\frac{k_\xi}{N}+\frac{\sqrt{\log(N)}}{\sqrt{NT}}+\gt\right),$$
\normalsize
with \small
$$\gt=(NT)^{2/\zeta}\left(\frac{1}{\sqrt{N}T}+\frac{1}{T^{3/2}}+(NT)^{2/\zeta}\frac{1}{T^2}\right).$$\normalsize
Furthermore, we have for $k \in \{0,1\}$
\small
\begin{equation*}
    \begin{aligned}
    &\left\|\frac{1}{T} \sum_{t=1}^T \BS w_t \BS \xi_{t-k}^\top\right\|_{\max}=O_P\left(\frac{{\log(N)}}{{T}}+\frac{k_\xi}{N}+\frac{\sqrt{\log(N)}}{\sqrt{NT}}+\gt\right)\\
    &\left\|\frac{1}{T} \sum_{t=1}^T \BS w_t \BS w_{t-k}^\top\right\|_{\max}=O_P\left(\frac{{\log(N)}}{{T}}+\frac{k_\xi}{N}+\frac{\sqrt{\log(N)}}{\sqrt{NT}}+\gt\right).
    \end{aligned}
\end{equation*}
\end{theorem}\normalsize

If $N=T^a$ and $a\leq \zeta-4$, we have $\gt\leq 1/\sqrt{NT}+1/T$ which means $\gt$ could be dropped in the above $O_p$ terms. 

We focus here on the lasso itself but the above theorem is also helpful for obtaining rates for the de-sparsified/de-biased lasso in this framework. As mentioned previously, if we just plug-in the rate for $\BS w_t$ we would obtain the slower rate of $O_P(\max(1/\sqrt{T},1/\sqrt{N}))$. 
With the results above we can establish bounds for the estimation error of the second step, as done in the following Theorem~\ref{thm.est.error.idio}.

\begin{theorem} \label{thm.est.error.idio}
Under Assumption \ref{sparsity}, \ref{ass.moments}, and \ref{ass.fac}
we have for $j=1,\dots,N$\small
\begin{align}
    \|\hat \bbeta^{(j)}- \bbeta^{(j)}\|_1&=O_P\Bigg(k\Bigg[\frac{\sqrt{\log(N)}}{\sqrt{T}}+\frac{(NpT)^{2/\zeta}}{T}+k\Bigg(\frac{k_\xi}{N}+\frac{\sqrt{\log(Np)}}{\sqrt{NT}}+(p)^{2/\zeta} \gt\Bigg)\Bigg]^{1-q}\Bigg) \label{eq.idio.est.error}
\end{align}
and
\begin{align}
&\|\hat \bbeta^{(j)}- \bbeta^{(j)}\|_2=O_P\Bigg(\sqrt{k}\Bigg[\frac{\sqrt{\log(N)}}{\sqrt{T}}+\frac{(NpT)^{2/\zeta}}{T}+k\Bigg(\frac{k_\xi}{N}+\frac{\log(Np)}{T}+\frac{\sqrt{\log(Np)}}{\sqrt{NT}}
+(p)^{2/\zeta} \gt\Bigg)\Bigg]^{1-q/2}\label{eq.idio.est.error.l2}\\
&+k^{3/2}\Bigg[\frac{\sqrt{\log(N)}}{\sqrt{T}}+\frac{(NpT)^{2/\zeta}}{T}+k\Bigg(\frac{k_\xi}{N}+\frac{\log(Np)}{T}+\frac{\sqrt{\log(Np)}}{\sqrt{NT}}+\gt\Bigg)\Bigg]^{(3-q)/2}\Bigg).\nonumber
\end{align}
\end{theorem} \normalsize
Recall that $q$ is the approximate sparsity parameter of Assumption~\ref{sparsity}. Let us have a closer look on the bound $\|\hat \bbeta^{(j)}- \bbeta^{(j)}\|_1$. Consider $N=T^a, p=T^b$ for some $a,b>0$ and let the following moment condition hold $\zeta\geq 4(1+a+b)$ as well as $k \leq \sqrt{N}$. Then, for $j=1,\dots,N$
 the bound simplifies to
 \small$$\|\hat \bbeta^{(j)}- \bbeta^{(j)}\|_1=O_P(k [\sqrt{\log(Np)/T}+k k_\xi/N]^{1-q}).$$\normalsize 
 The first condition, i.e., $k [\sqrt{\log(Np)/T}]^{1-q}=o(1)$, is standard for approximately sparse models, see among others Corollary 2.4 in \cite{van2016estimation}. The second condition, $k [k k_\xi/N]^{1-q}=o(1)$,  is not standard for approximately sparse models and appears due to the estimation error of the first step. That means the estimation error of the first step is negligible if (ignoring log terms for simplicity) $k k_\xi/N \leq 1/\sqrt T$. Hence, in the restrictive case of $k_\xi=O(1)$ the estimation error of the first stage is negligible if $N\geq k \sqrt{T}$. Let us also mention that without the detailed expression for $\BS w_t$, we would obtain the estimation error of the first step being of order $k/\sqrt{N}$. As $k_\xi$ is upper bounded by $\sqrt{N}$, the derived error bound with detailed expression for $\BS w_t$ is in no case less tight than the one without and it is tighter when $k_\xi/\sqrt{N}=o(1)$. 
 


\begin{remark}[\textit{Estimation with Strong Idiosyncratic Components}]
In the error bounds in Theorem~\ref{thm.est.error.idio}, the factor $k_\xi/N$ plays an important role. $k_\xi=\|\var(\BS \xi_t)\|_\infty$ quantifies the serial dependence of the idiosyncratic component. If this is large, the estimation in all steps suffers. Motivated by Generalized Least Squares (GLS), \cite{boivin2006more} proposes to weight the data such that the serial dependence of the idiosyncratic component can be decreased. This approach is also denoted \emph{generalized principal component analysis} and it is analyzed in more detail in \cite{choi2012efficient}. Let $\BS W \in \R^{N \times N}$ be a matrix of weights, then the factors are estimated using the weighted data $\BS X \BS W$. Note that we have $\var(\BS X \BS W)=\BS W \BS \Lambda \BS \Sigma_F \BS \Lambda^\top \BS W + \BS W \BS \Gamma_\xi(0) \BS W^\top$.  Hence, the factors can be estimated by a PCA of $\BS X \BS W$ whereas the loadings are obtained by regressing $\BS X$ onto the estimated factors. Since non-diagonal weighting schemes are seldom feasible without sparsity constraints, \cite{boivin2006more} suggest different diagonal weighting schemes. With the additional assumption that $\BS \Sigma_v$ 
is sparse, we suggest to use the VAR structure of the idiosyncratic component to obtain a more refined weighting scheme. To elaborate, we have that $\var(\BS \xi_t)=\BS \Gamma_\xi(0)=\sum_{j=0}^\infty \BS B^{(j)} \BS \Sigma_v (\BS B^{(j)})^\top$, where $(\BS B^{(j)})_{r,c}=(\A^j)_{r,c}, r,c=1,\dots,N$. Hence, $\BS \Gamma_\xi(0)$ is given by $\BS A^{(1)},\dots,\BS A^{(p)},\BS \Sigma_v$ and it can be estimated by plugging in estimators, see among others Theorem 5 in \cite{krampe2020statistical}. Let us denote this estimator as $\hat{\BS \Gamma}_\xi(0)$. Depending on whether sparsity constraints on $\BS \Sigma_v$ or $\BS \Sigma_v^{-1}$ are more realistic,  estimators are given by thresholding of the empirical covariance matrix \citep{bickel2008,cai2011adaptive} or by component-wise regularized regression \citep{friedman2008sparse,cai2016estimating,cai2016estimating2}. The weighting matrix is then given as $\BS W=\hat{\BS \Gamma}_\xi(0)^{-1/2}$. Consequently, the ``new'' $k_\xi$ is given by $\|\hat{\BS \Gamma}_\xi(0)^{-1/2} \BS \Gamma_\xi(0) \hat{\BS \Gamma}_\xi(0)^{-1/2}\|_\infty$ which can be considerably smaller if the employed estimators give reasonable results. Since the weighting leads also to a new estimation of the idiosyncratic component, it might be helpful to apply this approach more than once. \end{remark}


\begin{remark}[\textit{Similarity and Differences to Low-Rank plus Sparse Models}]\label{subsection.low-rank}
As mentioned in Section~1, the low-rank plus sparse VAR model discussed in \cite{basu2019low}, as well as the high-dimensional VAR models with strong cross-sectional correlated noise discussed in \cite{lin2019approximate,miao2022high}, are related to the model proposed here. We now stress the similarities and differences of these models starting with the model of \cite{basu2019low}. The low-rank plus sparse VAR model of order $p$ is given by 
$\BS x_t=\sum_{j=1}^p \BS \Theta^{(j)} \BS x_{t-j}+ \BS \eps_t$. The coefficient matrix can be decomposed as $\BS \Theta^{(j)}=\BS L^{(j)}+\BS S^{(j)}$ where $\bL^{(j)}$ is a low-rank matrix, $\BS S^{(j)}$ possesses some type of sparsity structure and $\BS \eps_t$ is some white-noise process.  The low-rank matrix takes here the role of the common component, see also \cite{bai2019rank}. Thus, this approach also combines a dense and a sparse approach. However, there are two major differences to the approach presented in this paper. First, note that while the low-rank plus sparse VAR model is some special form of a VAR$(p)$ model, a factor model with dynamic factors or idiosyncratic component is instead in general a VAR$(\infty)$ process even if the factors and idiosyncratic components follow finite order VAR processes. Second and most importantly, with the approach presented here we can derive estimation error bounds for a single time series, see Theorem~\ref{thm.prediction}. This is in contrast to the results derived in \cite{basu2019low}. They impose sparsity constraints on $\vecop(\BS S^{(j)})$\footnote{\cite{basu2019low} consider also a group-sparse structure for $\bS^{(j)}$. For this sparsity concept the discussion is quite similar.} and they do not estimate the VAR system row-wise as in \eqref{eq.lasso.beta.j}. Instead, all regression equations are combined using the Frobenius norm. The VAR slope matrices are considered as a sum of two matrices where the first matrix is regularized using the nuclear norm -- this imposes a low-rank structure -- and the second matrix is regularized using the $\ell_1$ norm on the vectorized matrix -- this imposes a sparse structure. They derive error bounds only regarding the Frobenius norm. That means they consider only the overall estimation error.  In connection with the sparsity constraints on $\vecop(\BS S^{(j)})$, this is too restrictive (or too less detailed) for the row-wise estimation error which is helpful for a forecast of a single time series. For a more detailed discussion of the different sparsity concepts and their implication regarding estimation error bounds, we refer to Section 2 in \cite{krampe2020statistical}. 

\cite{lin2019approximate,miao2022high} consider a model of the following form: $\BS x_t=\sum_{j=1}^p \BS A^{(j)} \BS x_{t-j}+ \BS \Lambda \BS f_t+\BS v_t$, where $\BS A^{(j)}$ are considered to be sparse matrices and $\BS \Lambda \BS f_t$ low-rank. This model is related in the following way to the model proposed here. A factor model $\BS{x}_t=\BS \Lambda \BS f_t+\BS{\xi}_t$ whose idiosyncratic component follows a VAR$(p)$ model, $\BS{\xi}_t=\sum_{j=1}^p \bA^{(j)}\BS{\xi}_{t-j}+\BS{v}_t$, can be written as $\BS{x}_t=\BS \Lambda \BS f_t-\sum_{j=1}^p \bA^{(j)} \BS \Lambda \BS f_{t-j}+\sum_{j=1}^p \bA^{(j)} \BS x_{t-j}+\BS{v}_t$. The component $\BS \Lambda \BS f_t-\sum_{j=1}^p \bA^{(j)} \BS \Lambda \BS f_{t-j}$ can be considered as the common component of a general dynamic factor model as in \cite{forni2000generalized} and it is low-rank. Hence, the model considered here and the model in \cite{lin2019approximate,miao2022high} differ in the low-rank component. 
Additionally, the sparsity assumptions on the slope matrices $\BS A^{(j)}$ and the estimation strategy differ. \cite{lin2019approximate} impose sparsity constraints on $\vecop(\BS A^{(j)})$. Furthermore, they combine all regression equations using the Frobenius norm and the low-rank part is handled by regularization of its nuclear norm. Similarly to \cite{basu2019low}, they derive error bounds only for the Frobenius norm. This means that for a forecast of a single time series the same drawbacks described above apply. 
\cite{miao2022high} use a three-step estimation procedure and they impose strict sparsity on the rows of $\bA^{(j)},j=1,\dots,p$. The first step is similar to the one in \cite{lin2019approximate}. The second and third estimation steps are used to refine the results. Especially for the second step, they use the estimated factor in a row-by-row regression. This enables them to obtain error bounds not only for the Frobenius-norm but also for $\|\cdot\|_\infty$-norm. 
\end{remark}

\section{Applications}\label{sec_applications}

\subsection{Forecasting and Factor-augmented regression} \label{subsec_forecasting}
Forecasting and factor-augmented regression is one of the most important uses of factor models, see among others \cite{stock2002forecasting}. Let us consider a $h$-step ahead forecast. Then, in factor-augmented regression, we have the
observables $\{\bX_t,\bY_t,\bW_t\}$ and the following regression model
\begin{align} \label{eq.fac.aug.reg}
\bY_{t+h}=\bbeta \bof_t+\balpha \bW_{t} + \boldsymbol{\eps}_{t+h}.    
\end{align}
Additionally, $\bX_t$ possesses the factor structure $\bX_t=\bLambda \bof_t+\bxi_t$ and
$\bW_t$ often consists of lagged values of $\bY_t$ and an intercept. To employ this regression model, $\bof_t$ needs to be estimated and the implications of using estimated regressors are analyzed among others in \cite{bai2006confidence,gonccalves2014bootstrapping}. $\bY_t$ can be here a scalar or vector and usually the regression model is estimated by least-squares which implies that  $\bW_t$ is considered as low-dimensional. 

In many applications, $\bY_t$ itself is a subset of $\bX_t$. Then, the model introduced in Section~\ref{sec_model} is an extension of the factor-augmented regression model. To elaborate, let 
$\bY_t$ be $N^Y$-dimensional and $\bI_{N,Y} \in \R^{N_Y\times N}$ be a selection matrix such that $\bY_t=\bI_{N,Y}^\top \bX_t$. Similarly, $\bW_t=\bI_{N,W}^\top \bX_t$, $\bI_{N,W} \in \R^{N_W \times N}$. For illustrative purposes, we focus on the case of one lag. Then, the factor-augmented regression model \eqref{eq.fac.aug.reg} reads as 
$$
\bY_{t+h}=\bI_{N,Y}^\top \bX_{t+h}=\bI_{N,Y}^\top (\bB \bof_t+\bA \bI_{N,W} \bI_{N,W}^\top \bX_{t} + \boeta_{t+h}),
$$
where $\bbeta=\bI_{N,Y}^\top \bB$, 
$\balpha=\bI_{N,Y}^\top \bA \bI_{N,W}$, and $\boldsymbol{\eps}_{t+h}=\bI_{N,Y} \boeta_{t+h}$. In this set-up the factor-augmented regression can be understood as a feasible (and non-sparse) approximation of the regression of $\BS Y_{t+h}$ onto $\BS X_t$. 
With this motivation in mind, a natural and feasible extension of the previous model would be 
\begin{align} \label{eq.fac.aug.reg.sparse1}
\bY_{t+h}=\bI_{N,Y}^\top \bX_{t+h}=\bI_{N,Y}^\top(\bB \bof_t+\tilde \bA \bX_{t} + \boeta_{t+h}),    
\end{align}
where $\tilde \bA$ is considered as sparse. The idea behind this extension is that instead of doing the selection by hand, i.e, choosing $\bI_{N,W}$, the selection is done automatically by a data-driven selection procedure such as lasso. Since $\bX_t$ decomposes into factor and idiosyncratic part, also $\boeta_t$ should decompose into such parts. Let $\boeta_{t+h}=\bLambda \bu_{t+h}+\bv_{t+h}$ and $\bu_{t+h}=\bof_{t+h}-\bD \bof_t$, $\bv_{t+h}=\bxi_{t+h}-\bE \bxi_{t}$. Then, we have
\begin{align*}
\bY_{t+h}=\bI_{N,Y}^\top \bX_{t+h}=
\bI_{N,Y}^\top(\bB \bof_t+\tilde \bA \bX_{t} + \bLambda \bof_{t+h}-\bLambda \bD \bof_t+\bxi_{t+h}-\bE \bxi_t),
\end{align*}
which is equivalent to
$$
0=\bI_{N,Y}^\top((\bB+\tilde \bA \bLambda-\bLambda \bD)\bof_t+(\tilde \bA-\bE) \bxi_t).
$$
If $\bof_t$ and $\bxi_t$ are uncorrelated, we have $\tilde \bA=\bE$ and $\bB=\bLambda \bD-\tilde \bA \bLambda$. Thus, the sparse extended factor-augmented model \eqref{eq.fac.aug.reg.sparse1} reads as the following state-space model, which is a special case of the model described in Section~\ref{sec_model}
\begin{align}
\bY_{t+h}&=\bI_{N,Y}^\top \bX_{t+h}, \nonumber \\
\bX_{t+h}&= \bLambda \bof_{t+h} + \bxi_{t+h},\nonumber \\
\bof_{t+h}&=\bD \bof_t+\bu_{t+h}  \label{eq.fac.aug.state.space}, \\
\bxi_{t+h}&=\tilde \bA \bxi_{t} + \bv_{t+h}. \nonumber
\end{align}

With this connection in mind, we see that a forecast of $\bX_t$ is built upon forecasting $\bof_t$ and $\bxi_t$. We are now going to present the forecast method for a one-step-ahead prediction. 
An $h$-step-ahead prediction can be done recursively. If the main interest is on an $h$-step-ahead forecast, a direct $h$-step-ahead forecast can be more accurate, see among others \cite{smeekes2018macroeconomic}. A direct $h$-step-ahead forecast can be obtained by changing the regression equation from $t$ to $t+h$ as in \eqref{eq.fac.aug.state.space}. 
Based on the estimation method proposed in the previous section, the approach is as follows.  First, a standard linear one-step-ahead prediction is computed based on the estimated factors. Combining this prediction with the loadings gives a prediction of the common component. Second, the sparse VAR model is used to get a prediction of the idiosyncratic component. Finally, the sum of these two predicted components gives the prediction of the original process $\BS x_t$. To elaborate, consider first that the factors and idiosyncratic component are observed. Then, let 
$\BS f^{(1,p_f)}_{T+1}=\sum_{j=1}^{p_f} {\BS \Pi}_j^{(p_f)} {\BS f}_{T+1-j}$ be the linear one-step-ahead prediction  based on ${\BS f}_{T},\dots,{\BS f}_{T-p_f}$, where $\sum_{j=1}^{p_f} {\BS \Pi}_j^{(p_f)} {\BS \Gamma}_{f}(i-j)={\BS \Gamma}_f(i),i=1,\dots,p_f$ and $ {\BS \Gamma}_f(i-j)=\E {\BS f}_{t+i} {\BS f}_{t-j}^\top$, see among others Section 11.4 in \cite{BrockwellDavis1991}. Furthermore, since $\{ \BS \xi_t\}$ follows a VAR$(p)$ model, ${\BS \xi}_{T+1}^{(1)}=\sum_{j=1}^p {\BS A^{(j)}} {\BS \xi}_{T-j}$ is the one-step-ahead prediction for the idiosyncratic component. That means, $\BS X^{(1,p_f)}_{T+1}=\BS \Lambda \BS f^{(1,p_f)}_{T+1}+{\BS \xi}_{T+1}^{(1)}$ is the joint one-step-ahead prediction for $\BS X_{T+1}$ with the prediction error 
$\var(\BS X_{T+1}-\BS X_{T+1}^{(1,p_f)})=\BS \Lambda \var(\BS f^{(1,p_f)}_{T+1}-\BS f_{T+1}) \BS \Lambda^\top +\bSigma_v$ and for a single variable $j$ we have $\var(\BS e_j^\top (\BS X_{T+1}-\BS X_{T+1}^{(1,p_f)}))=\BS \Lambda_j^\top \var(\BS f^{(1,p_f)}_{T+1}-\BS f_{T+1}) \BS \Lambda_j +\be_j^\top\bSigma_v \be_j$. If $\{\BS f_t\}$ follows a VAR$(p_f)$ model, this simplifies to $\var(\BS X_{T+1}-\BS X_{T+1}^{(1,p_f)})=\BS \Lambda \bSigma_u \BS \Lambda^\top +\bSigma_v$.


Since the parameters are unknown and the factors and idiosyncratic component are latent, this approach is unfeasible but the results of Theorem~\ref{thm.est.error.step.one} and \ref{thm.est.error.idio} help to obtain a feasible approach.
For this, we construct feasible counterparts of the prediction approach above. Let  $\hat {\BS f}_{T+1}^{(1,p_f)}=\sum_{j=1}^{p_f} \hat {\BS \Pi}_j^{(p_f)} \hat {\BS f}_{T+1-j}$ be the linear one-step-ahead prediction  based on $\hat {\BS f}_{T},\dots,\hat {\BS f}_{T-p_f}$, where $\sum_{j=1}^{p_f} \hat {\BS \Pi}_j^{(p_f)} \hat {\BS \Gamma}_{f}(i-j)=\hat {\BS \Gamma}_f(i),i=1,\dots,p_f$, and $\hat {\BS \Gamma}_f(i-j)=1/n \sum_{t=1+j}^{T-i} \hat {\BS f}_{t+i} \hat {\BS f}_{t-j}^\top$. Furthermore, let $\hat {\BS \xi}_{T+1}^{(1)}=\hat {\BS A} (\hat {\BS \xi}_{T}^\top,\dots,\hat {\BS \xi}_{T-p}^\top)^\top$ be the one-step-ahead prediction for the idiosyncratic component. Then, $\hat \bX_{T+1}^{(1,p_f)}=\hat {\BS \Lambda} \hat {\BS f}_{T+1}^{(1,p_f)}+\hat {\BS \xi}_{T+1}^{(1)}$ is the joint and feasible one-step-ahead prediction for $\BS X_{T+1}$. Even though a high-dimensional time series system is considered, the interest is often in the prediction of some key time series. We quantify in the following Theorem~\ref{thm.prediction} the estimation error between the feasible and unfeasible approach for a single time series.

\begin{theorem} \label{thm.prediction}
Under Assumption \ref{sparsity}, \ref{ass.moments}, and \ref{ass.fac}
we have for $j=1,\dots,N$ \small
\begin{align*}
&\be_j^\top(\hat {\BS X}_{T+1}^{(1,p_f)}-{\BS X}_{T+1}^{(1,p_f)})=O_P\Bigg(\frac{1}{\sqrt{N}}+k\Bigg[\frac{\sqrt{\log(Np)}}{\sqrt{T}}+\frac{(NpT)^{2/\zeta}}{T}+k\Bigg(\frac{k_\xi}{N}+\frac{\sqrt{\log(Np)}}{\sqrt{NT}}+(p)^{2/\zeta}\gt \Bigg)\Bigg]^{1-q}\Bigg).
\end{align*}
\end{theorem}\normalsize
In relation to the error bound for $\|\hat \bbeta^{(j)}- \bbeta^{(j)}\|_1$ derived in Theorem~\ref{thm.est.error.idio} only an additional $1/\sqrt{N}$ appears which arises due to the estimation of the factors. 

Let us mention that $\bX_t$ does not need to consist of series at the same time point. For instance, predicting inflation and GDP at time $t$ with other variables available up to time $t$, the vector $\bX_t$ can be build as inflation$(t-1)$, GDP$(t-1)$ and other variables$(t)$. Then, using the proposed approach to obtain a prediction of $\bX_{t+1}$ gives a prediction of inflation and GDP at time point $t$.

\subsection{Bootstrap of factor models}
{In factor-augmented regression, when inference for the regression coefficients of the factors is of interest, the estimation of the factors needs to be taken into account. In cases when $N$ is large in relation to $T$, the estimated factors can be treated as observed, see \cite{bai2006confidence} for details. However, if $\sqrt{T}/N\to c, c>0$ some bias term appears which contains among others $\mathbb{E} \bxi_t \bxi_t=\BS \Gamma_\xi(0)$, see \cite{gonccalves2014bootstrapping}. To assess this bias term,  \cite{gonccalves2020bootstrapping} propose a bootstrap algorithm that mimics the cross-sectional dependence structure of the idiosyncratic component. The bootstrap relies on an estimate of $\BS \Gamma_\xi$. For this, they assume sparsity of $\BS \Gamma_\xi(0)$ and estimate it via thresholding. If $\bxi_t$ is driven by a sparse VAR model, assuming sparsity on $\BS \Gamma_\xi(0)$ restricts (in a not traceable way) the sparsity of the slope parameter of the VAR model, see \cite{krampe2018bootstrap}. To avoid this, two options exist and both rely also on an estimate of the variance matrix of $\bv_t$, the innovations of the VAR process. First, 
$\BS \Gamma_\xi(0)$ can be estimated using the VAR structure, see estimator (6) in \cite{krampe2018bootstrap}. With this estimated $\BS \Gamma_\xi(0)$, the bootstrap approach of \cite{gonccalves2020bootstrapping} can be used. Second,
one can extend the bootstrap approach of \cite{gonccalves2020bootstrapping} and mimic not only the contemporaneous dependence structure of $\bxi_t$ but also the entire second-order structure of $\bxi_t$ by leveraging on the sparse VAR structure. For this, one can follow the bootstrap algorithm of \cite{krampe2018bootstrap} (specifically their step~1 and step~2 in Section~3). Since this also mimics the dependence over time, it could improve finite sample performance. Furthermore, this extended bootstrap approach can also be used to obtain inference results for the loadings. The asymptotic normality of the loadings is derived in \cite{bai2003} and for $\BS \Lambda_i$ the asymptotic variance contains among others terms such as $1/T \sum_{s=1,t=1}^T \mathbb{E} \bof_t \bof_s^\top \bxi_{i,s} \bxi_{i,t}$. Under independence of factors and idiosyncratic component, this simplifies to $1/T \sum_{s=1,t=1}^T \BS \Gamma_f(t-s) \be_i^\top \BS \Gamma_{\xi}(t-s) \be_i$ and only second-order moments of the idiosyncratic component appear. Hence, a successful bootstrap approach for the loadings does not only need to mimic the dependence structure of the factors but also that of the autocovariance of the idiosyncratic component which our proposed extension achieves.}


\subsection{Estimation of time series dependence networks}
\label{sec_spec_dens}
{For multivariate time series, networks are often used to display the connection structure between the individual time series. In these networks, each time series is represented by a node and an edge between two nodes is drawn if some form of connection between the two time series exists. Here, several approaches are available to define a connection. In the context of (Gaussian) graphical models, a connection is drawn based on the partial correlation structure, which translates in the time series context to the partial coherence structure, see \cite{brillinger1996remarks,dahlhaus2000graphical}. Other approaches to defining a connection are based on Granger-causality \citep{granger1969investigating,hecq2021granger} and forecast error variance decompositions, see \cite{diebold2014network}. In the following, we elaborate on the connections based on partial coherences and graphical models. The partial coherence measures the strength of the linear relations between two time series after eliminating all indirect linear effects caused by all other time series of the system, taking into account all leads and lags relations.}

{ To elaborate,  consider two components $u$ and $v$, then the partial coherence at frequency $\omega \in [0,2\pi]$ is given by
\begin{align} \label{eq.partial.coherence}
    R_{u,v}(\omega)=|\rho_{u,v}(\omega)|, \ \ \mbox{where} \ \ \rho_{u,v}(\omega)=-f_{u,v}^{-1}(\omega)/\sqrt{f_{u,u}^{-1}(\omega)f_{v,v}^{-1}(\omega)},
\end{align}
where $f_{u,v}^{-1}(\omega)$ denotes the $(u,v)$th element of the inverse of the spectral density matrix at frequency $\omega$. An edge is drawn between component $u$ and $v$ if $\sup_\omega R_{u,v}(\omega)\geq \delta$ for some $\delta \in [0,1)$. $\delta$ is a user-specified threshold determining which connections are important. Note that $\delta=0$ includes all non-zero connections. However, in the presence of a factor, it is most likely that $ \sup_\omega R_{u,v}(\omega)>0$ for all $u,v$ and it is of more interest to identify those which exceed some positive threshold. To inherit such a network from data, the spectral density needs to be estimated.}
When the dimension of the time series is small, the spectral density matrix is often estimated by non-parametric approaches as lag-window estimators or smoothed periodograms, respectively, see among others \cite{brillinger2001time,koopmans1995spectral,hannan2009multiple,wu2018asymptotic}. In a high-dimensional set-up, the problem of estimating the spectral density matrix or its inverse has been extensively investigated in the literature during the last decade. One approach is to combine the non-parametric lag-window estimators with regularization techniques developed for the covariance and precision matrix estimation, see among others \cite{sun2018large,fiecas2019spectral,zhang2020convergence}. Such approaches work under the assumption that the spectral density matrix or its inverse is sparse. However, a direct sparsity assumption on the spectral density matrix or its inverse is contradicting the assumption of the existence of factors. That means the factors need to be taken into account in the estimation of the spectral density matrix. The procedure developed in the previous section can be used to obtain (under slightly modified assumptions) a consistent estimator of the inverse of the spectral density matrix. Since the VAR structure of the idiosyncratic component is used, we obtain a semiparametric estimator for the inverse of the spectral density matrix. 
Let us mention that the factors can be also taken into account by using a low-rank plus sparse approach applied to a smooth periodogram, see \cite{barigozzi2021algebraic}. This estimator differs, however, from the one presented here in several aspects. First, the low-rank plus sparse approach describes in finite samples a different model than the approach used here, see also the discussion of low-rank plus sparse structures in Remark~\ref{subsection.low-rank}. Second, they focus on consistency results regarding $\|\cdot\|_2$ whereas we present here also row- and column-wise consistency results i.e., consistency with respect to $\|\cdot\|_1$ and $\|\cdot\|_\infty$. 

Let us begin with defining the spectral density matrix of the time series $\{\BS X_t\}$ given by \eqref{eq.fac.decom}. The spectral density matrix of the factor process specified in Assumption~\ref{ass.moments} is given by \small
$$
\BS f_f(\omega)= \bigg[\sum_{j=0}^\infty \BS D^{(j)} \exp(-i j \omega)\bigg] \BS \Sigma_u \bigg[\sum_{j=0}^\infty \BS D^{(j)} \exp(i j \omega)\bigg]^\top, \quad \omega \in [0,2\pi]
$$\normalsize
and for the idiosyncratic component driven by a VAR$(p)$ we have \small
$$
\BS f_\xi(\omega)=\bigg[\BS I_N-\sum_{j=1}^p \BS A^{(j)} \exp(-i j \omega)\bigg]^{-1} \BS \Sigma_v \bigg(\bigg[\BS I_N-\sum_{j=1}^p \BS A^{(j)} \exp(i j \omega)\bigg]^{-1}\bigg)^\top,
$$\normalsize
with the inverse \small
$$
\BS f_\xi(\omega)^{-1}=\bigg[\BS I_N-\sum_{j=1}^p \BS A^{(j)} \exp(i j \omega)\bigg]^\top \BS \Sigma_v^{-1} \bigg[\BS I_N-\sum_{j=1}^p \BS A^{(j)} \exp(-i j \omega)\bigg].
$$\normalsize
That is, the spectral density of the process $\{\BS X_t\}$ is given by \small
\begin{align}
    \label{eq.spectral.density.X}
    \BS f_X(\omega)=\BS \Lambda \BS f_f(\omega) \BS \Lambda^\top +\BS f_\xi(\omega),
\end{align}\normalsize
and its inverse using the Sherman–Morrison–Woodbury formula is given by\small
\begin{align}
    \label{eq.spectral.density.X.inverse}
    \BS f_X(\omega)^{-1}=\BS f_\xi^{-1}(\omega)-\BS f_\xi^{-1}(\omega) \BS \Lambda \bigg[\BS f_f^{-1}(\omega)+\BS \Lambda^\top \BS f_\xi^{-1}(\omega) \BS \Lambda \bigg]^{-1} \BS \Lambda^\top \BS f_\xi^{-1}(\omega).
\end{align}\normalsize

We estimate $\BS f_X(\omega)^{-1}$ by estimating $\BS f_f^{-1}$ and $\BS f_\xi^{-1}$ separately. Note that the factors lead to an unbounded $\|\BS f_X(\omega)\|_2$ for growing dimension but 
the inverse is stable, i.e., $\|\BS f_X(\omega)^{-1}\|_2$ is bounded.

As it is of fixed dimension $r$, the spectral density $\BS f_f$ or its inverse can be estimated by classical methods such as non-parametric lag-window estimators. For this, let $K$ be a kernel fulfilling Assumption~1 in \cite{wu2018asymptotic}. That is, $K$ is an even and bounded function with bounded support in $(-1,1)$, continuous in $(-1,1)$, $K(0)=1, \kappa=\int_{-1}^1 K^2(u) du<1$, and $\sum_{l \in \Z} \sup_{|s-l|<1} |K(l\omega)-K(s\omega)|=O(1)$ as $\omega \to 0$. Furthermore, let $B_T=T^b, b \in (0,1)$ be the lag-window size fulfilling Assumption~2 in \cite{wu2018asymptotic}. Then, a spectral density estimator is given by \small
\begin{align}
    \hat{ \BS f}_f(\omega)=\frac{1}{2\pi}\sum_{h=-T+1}^{T-1} K\left(\frac{h}{B_T}\right) \exp(-ih \omega) \hat {\BS \Gamma}_f(h),
\end{align}\normalsize
where $\hat{ \BS \Gamma}_f(h)$ is the sample autocovariance function $\hat { \BS\Gamma}_f(h)=1/T\sum_{t} \hat{\BS f}_{t+h} \hat{ \BS f}_t^T$. Based on observations $\BS f_1,\dots,\BS f_T$, let $\tilde {\BS f}_f(\omega)=\frac{1}{2\pi}\sum_{h=-T+1}^{T-1} K\left(\frac{u}{B_T}\right) \exp(-ih \omega) \tilde{ \BS \Gamma}_f(h),\; \tilde{\BS \Gamma}_{f} (h)=1/T\sum_{t} \BS f_{t+h} \BS f_t^T,$ be the (unfeasible) estimator of $\BS f_f$. Then, the results of \cite{wu2018asymptotic} give that $\|\tilde {\BS f}_f(\omega)-\BS f_f(\omega)\|_{\max}=O_P(\sqrt{B_T/T})$. With this result and noting that $\{\BS f_t\}$ is a process of fixed dimension $r$, consistency of $\hat {\BS f}_f(\omega)$ follows by Lemma~\ref{lemma.1}, see Lemma~\ref{lemma.spectral.density.factor} in the appendix for details.

As mentioned, we use the VAR structure of the idiosyncratic component to estimate its spectral density matrix. In the previous section, we showed that the VAR parameters of the idiosyncratic component can be estimated row-wise consistently, i.e., consistency of $\A$ for the matrix norm $\|\cdot \|_\infty$. However, the estimation of the spectral density requires additional column-wise consistency, that is consistency of $\BS A^{(j)},j=1,\dots,p,$ with respect to $\|\cdot\|_1$. Such a column-wise consistency requires additional sparsity assumptions, see also \cite{krampe2020statistical} for a discussion. Furthermore, a parametric estimation of the spectral density matrix of a VAR process requires an estimate of the covariance or precision matrix of the residual process $\{\BS v_t\}$. Since our focus is on the estimation of the inverse of the spectral density matrix, we estimate the precision matrix and formulate sparsity assumption on this matrix. 
See Assumption~\ref{sparsity.b} for the exact definition of the additional sparsity assumptions.
\begin{assumption}\label{sparsity.b}(\textit{Sparsity and stability})\\
$(i)$ 
The VAR process is row- and column-wise approximately sparse with approximate sparsity parameter $q \in [0,1)$, i.e., \small
$$\sum_{l=1}^p \max_i \sum_{j=1}^N |\BS A_{i,j}^{(l)}|^q\leq k, \qquad \sum_{l=1}^p \max_j \sum_{j=i}^N |\BS A_{i,j}^{(l)}|^q\leq  k.$$\normalsize
$(ii)$ As in Assumption~\ref{sparsity} (ii) and $\sup_\omega \|\BS f _\xi(\omega)\|_\infty \leq k_\xi M$.\\
$(iii)$ The precision matrix $\BS \Sigma_v^{-1}=\var(\BS v_t)^{-1}$ of the VAR innovations $\{\BS v_t\}$ is positive definite and approximately sparse and $\|\BS \Sigma_v^{-1}\|_2\leq M$. Let $q_v \in[0,1)$ denote the approximate sparsity parameter and $k_v$ the sparsity. Then, \small
$$
\max_i \sum_{j=1}^N |(\BS \Sigma_v^{-1})_{i,j}|^{q_v}=\max_j \sum_{i=1}^N |(\BS \Sigma_v^{-1})_{i,j}|^{q_v}\leq k_v. 
$$ \normalsize
\end{assumption}

As mentioned, the precision matrix of the residuals $\{\BS v_t\}$ needs to be estimated. The residuals can be estimated by $\hat {\BS v}_t=\hat {\BS \xi}_t-\sum_{j=1}^p \hat {\BS A}^{(j)} \hat {\BS \xi}_{t-j}, t=p+1,\dots,T$. Then, based on these estimated residuals, procedures like graphical lasso of \cite{friedman2008sparse} or (A)CLIME of \cite{cai2011constrained,cai2016estimating,cai2016estimating2} can be used. In the proofs we consider the CLIME method and denote this estimator by $\hat {\BS \Sigma}_v^{-1,CLIME}$ but similar results can be established for the graphical lasso estimator. Then, we construct the following estimator for $\BS f_\xi^{-1}(\omega)$ 
\small
\begin{align}
    \hat{\BS f}_\xi(\omega)^{-1}=\bigg[\BS I_N-\sum_{j=1}^p \hat{\BS A}^{(thr,j)} \exp(i j \omega)\bigg]^\top \hat {\BS \Sigma}_v^{-1,CLIME} \bigg[\BS I_N-\sum_{j=1}^p \bA^{(thr,j)} \exp(-i j \omega)\bigg],
\end{align}\normalsize
where $\hat{\BS A}^{(thr,j)}=(\THRarg{\lambda_\xi}\hat{\BS A}^{(j)})$ and $\THRarg{\lambda_\xi}$ is a thresholding function with threshold parameter $\lambda_\xi$ fulfilling the conditions $(i)$ to $(iii)$ in Section 2 in \cite{cai2011adaptive}. For instance, such a thresholding  function can be the adaptive lasso thresholding function given by $\THRarg{\lambda_\xi}^{al}(z)=z(1-|{\lambda}/z|^\nu)_+$ with $\nu\geq1$. Soft thresholding ($\nu=1$) and hard thresholding ($\nu=\infty$) are boundary cases of this function. This thresholding functions act by thresholding every element of the matrix $\hat{\BS A}^{(j)}$ and it results in a row- and column-wise consistent estimation of the VAR slope matrices. In Lemma~\ref{lem.rate.spectral.var} in the appendix, we present the error bounds $\|\hat{\BS f}_\xi^{-1}(\omega)-{\BS f}_\xi^{-1}(\omega)\|_\infty$ and $\|\hat{\BS f}_\xi^{-1}(\omega)-{\BS f}_\xi^{-1}(\omega)\|_2$. Finally, replacing in \eqref{eq.spectral.density.X.inverse} all quantities with the estimators discussed above leads to our final estimator of the inverse of the spectral density matrix of $\{\BS X_t\}$. Its error bounds are given in the following Theorem~\ref{thm.spectral.density}. We only present here explicitly the rate for a simplified case. In the general case, an explicit rate can be obtained by inserting the results of Lemma~\ref{lem.rate.spectral.var} and Theorem~\ref{thm.est.error.step.one}. Since it leads to a lengthy and not insightful expression, we omit it here. The rate is dominated by the estimation error of the sparse VAR and it is similar to the one in Theorem~\ref{thm.est.error.idio}. However, the rate is more affected by the sparsity parameter in the sense that its maximum growth rate is less for the spectral density than it is for prediction. Maximum growth rate refers here to the maximal rate of sparsity for which consistency can be achieved. 

\begin{theorem}\label{thm.spectral.density}
Under Assumption~\ref{ass.moments},\ref{ass.fac},\ref{sparsity.b} and Assumption~1 and 2 in \cite{wu2018asymptotic} (conditions on the used kernel and lag-window of the non-parametric estimator) we  have the following \small
\begin{align*}
    \|\BS f_X(\omega)^{-1}-\hat{\BS f}_X(\omega)^{-1}\|_l=O_P(k_\xi\|\hat{\BS f}_\xi^{-1}(\omega)-{\BS f}_\xi^{-1}(\omega)\|_\infty+k_\xi^2 \|\hat {\BS  \Lambda}-\BS \Lambda \bH_{NT}^{-1}\|_{\max}), l \in [1,\infty]
\end{align*}\normalsize
and  \small
\begin{align*}
    \|\BS f_X(\omega)^{-1}-\hat {\BS f}_X(\omega)^{-1}\|_2=O_P(\|\hat{\BS f}_\xi^{-1}(\omega)-{\BS f}_\xi^{-1}(\omega)\|_2+\|\hat {\BS \Lambda}-\BS \Lambda \bH_{NT}^{-1}\|_{\max}).
\end{align*}
\normalsize
If $N=T^a, p=T^b$ for some $a,b>0$, $\zeta\geq 4(1+a+b)$ and $k=o(\sqrt{T/\log(Np)})$, these error bounds simplify to 
\small
\begin{align*}
    \|\BS f_X(\omega)^{-1}-\hat {\BS f}_X(\omega)^{-1}\|_l=O_P\Bigg(& k^2 \|{\BS \Sigma}_v^{-1}\|_1 \Big(k_v  \Big[\sqrt{(\log(N)/T})+k\Big[\frac{k_\xi}{N}+\frac{\log(N)}{T}+\frac{\sqrt{\log(N)}}{\sqrt{NT}}\Big]\Big]^{1-q_v}+\\
    & \sqrt{k} \Big[\frac{\sqrt{\log(N)}}{\sqrt{T}}+k k_\xi/N+k \frac{\sqrt{\log(N)}}{\sqrt{NT}}\Big]^{1-q/2} \Big)\Bigg), l \in [1,\infty],
\end{align*}
\begin{align*}
    \|\BS f_X(\omega)^{-1}-\hat {\BS f}_X(\omega)^{-1}\|_2&=O_P\Bigg(k_v \|{\BS \Sigma}_v^{-1}\|_1 \Bigg[\sqrt{(\log(N)/T})+k\Big[\frac{k_\xi}{N}+\frac{\log(N)}{T}+\frac{\sqrt{\log(N)}}{\sqrt{NT}}\Big]\Bigg]^{1-q_v}\\&+k^{3/2} \Big[\frac{\sqrt{\log(N)}}{\sqrt{T}}+k k_\xi/N+k \frac{\sqrt{\log(N)}}{\sqrt{NT}}\Big]^{1-q/2} \Bigg).
\end{align*}
\end{theorem}\normalsize

\begin{example}
Let us showcase an example of 
partial coherence network construction using the FRED-MD dataset \citep{mccracken2016fred} which contains a large number of U.S. macroeconomic series sampled at monthly frequency. After the necessary cleaning of the data set due to missings, we are left with $123$ macroeconomic series for a time span ranging from January $1959$ until December $2019$.\footnote{We intentionally truncate the last few years to exclude the Covid-19 crisis.} We base the analysis on the partial coherence in \eqref{eq.partial.coherence} computed from the estimated inverse spectral density matrix with our proposed factor model with sparse VAR idiosyncratic components, as described in Section \ref{sec_spec_dens}. We determine the number of factors using the criteria of \cite{bai2002determining} with the penalty function $g(N,T)=(N+T)/(NT) \log(NT/(N+T))$. The lag-length of the sparse VAR for the idiosyncratic component is selected by the information criteria \eqref{crit_FactLags_global} discussed in depth in the next section and directly applied to the estimated idiosyncratic component with $r_{\max}=0$. We consider the two halves of the sample 1959-2019, namely January 1960 until December 1989 and January 1990 until December 2019. We consider a lower-bound level of partial coherence of $R>0.05$. 
In the figures, the labels of the macroeconomic variables are accompanied by a number within square brackets which refers to the group they belong to according to FRED-MD.\footnote{ Group 1 is ``Output and Income", group 2 is ``Labor Market", group 3 is ``Housing", group 4 is ``Consumption, Orders and Inventories". group 5 is ``Money and Credit", group 6 is ``Interest and Exchange rates", group 7 is ``Prices" and finally group 8 is ``Stock Market".} Active vertices i.e., vertices that are connected with at least one of the others, are reported in red while the non-active ones are in light blue.
\begin{figure}[tb]
  \centering
\subfloat[1960-1989\label{fig:test2}]
  {\includegraphics[width=.45\linewidth, trim={1.5 1.5 1.5 1.5},clip]{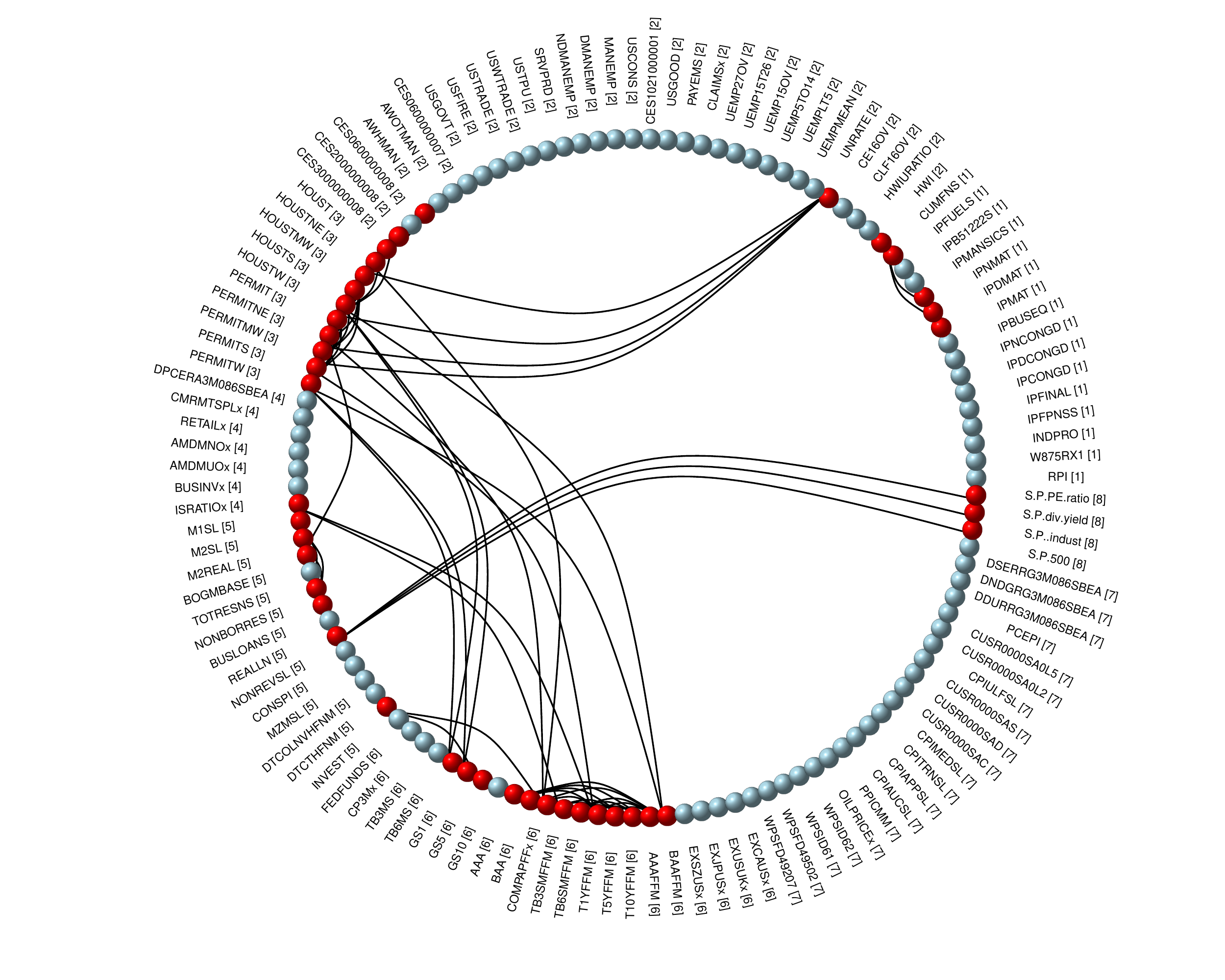}}\hfill 
\subfloat[1990-2019\label{fig:test3}]
  {\includegraphics[width=.45\linewidth, trim={1.5 1.5 1.5 1.5},clip]{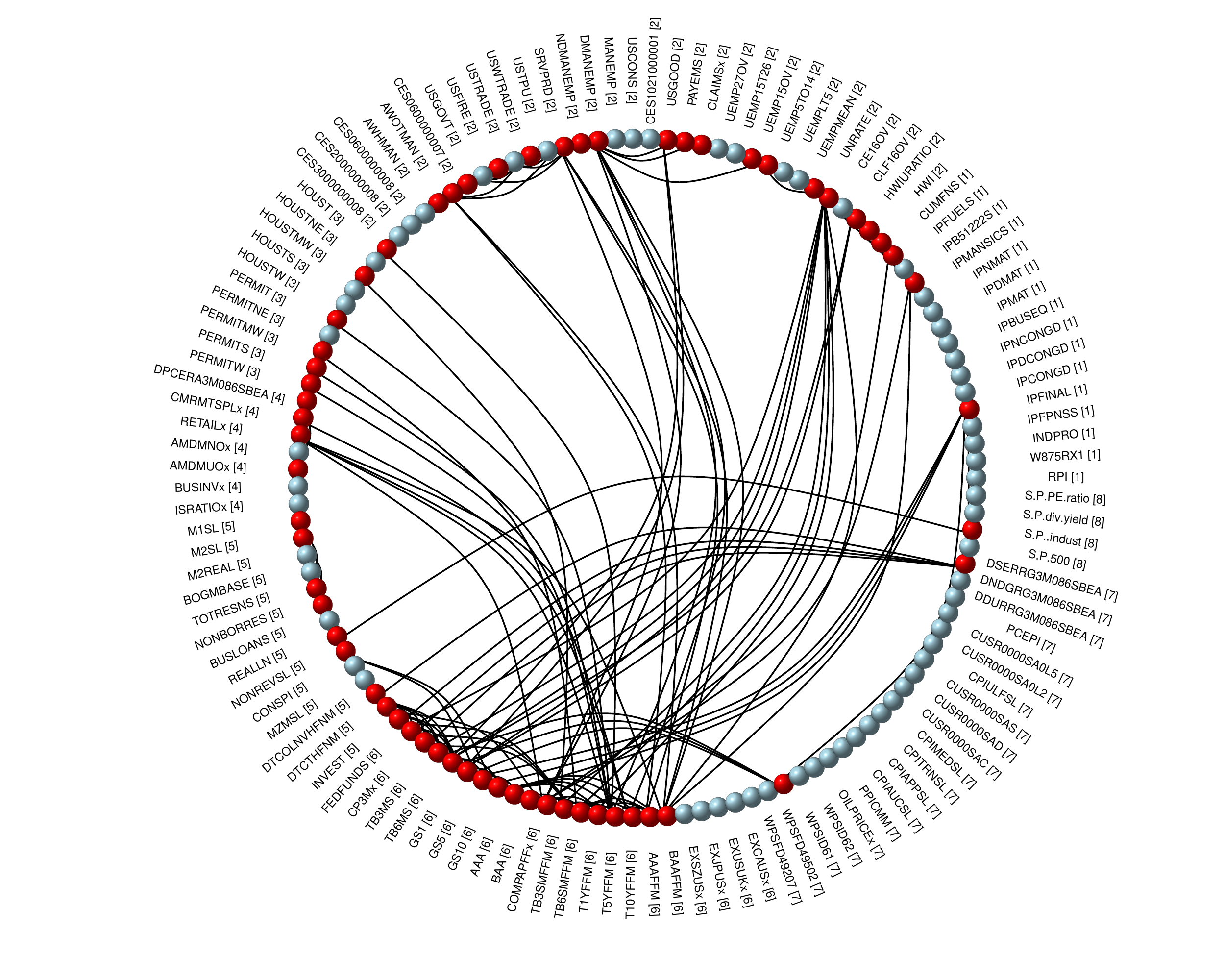}}
\caption{Partial coherence networks, $R>0.05$. [zoom in possible in pdf-version]}\label{pc005}
\end{figure}
In Figure \ref{pc005}, the highest number of connections is observed in the second half of the sample, in panel (\ref{fig:test3}). $60$ active vertices are found compared to the $42$ active in the first half of the sample in panel (\ref{fig:test2}). 
``Interest and Exchange Rates" group 6 is the most active group of vertices across the sub-samples: $17$ of its variables are connected in the second half of the sample, while $13$  are active in the first half of the sample. Group 2, 3 and 5 i.e., respectively: ``Labor Market", ``Housing" and ``Money and Credit" are also particularly active. In fact, in the second half of the sample, $17$ vertices belonging Labor Market are found, compared to only $3$ in the first half. $10$ active vertices within Housing are found in the first half of the sample compared to $6$ in the second half. Vertices belonging to Money and Credit and Prices are $7$ for the first half of the sample and $9$ for the second half. 
\end{example}

\section{{Joint Selection: Number of Factors \& Lag-length}}\label{sec_numbfactors}

{The present context clearly requires the selection of the number of factors within the PCA step as well as the order of the VAR for the idiosyncratic component. In the literature, there is an abundance of methods available for both. Among others, the seminal work of \citet{bai2002determining} introduced information criteria for a data-driven specification of the number of factors and it is perhaps the most employed method in practice. Further refinements of this method can be found in \citet{hallin2007determining,alessi2010improved}. Information criteria can also be used to specify the order of a VAR. For instance, \citet{hecq2021granger} propose to marginalize the (high-dimensional) VAR into a sequence of AR(p) regressions and select the lag-length via an approximated Bayesian information criterion (BIC). The consistency of the BIC has been proved in \citet{wang2009shrinkage}. Under a few technical conditions on the divergence speed of the model dimension and the size of non-zero coefficients, they show how a slightly modified BIC can identify the true model consistently even when the dimension diverges. }
\par We propose here a unified procedure able at the same time to consistently estimate the lag-length as well as the number of factors. For a given lag-length and number of factors, the penalty parameter can {also} be chosen with {an information criterion as AIC or BIC but this necessarily needs to be distinct from the joint information criteria for the number of factors and lag-length hence we briefly discuss it first. Let $\bxi_{t,S}^v$ be the subvector containing those columns of $\bxi_t^v$ belonging to the set $S$. Let further $\hat{S}$ be the active set identified by the lasso for a given $\lambda$. Then the value $\lambda^{IC}$ chosen by information criteria is found as \small 
\begin{equation*}
\lambda^{IC} = \underset{\lambda}{\arg\min} \left(\ln \left(\frac{1}{T-p+1}\sum_{t=p+1}^T\left(\xi_{j,t}-\sum_{j=1}^p \bbeta^\top_{S(\lambda)}\bxi_{t-j,S(\lambda)}^{v}\right)^2\right)+\left(\frac{1}{T-p+1}\right)C_T df\right),
\end{equation*} \normalsize
where $df$ represents the degrees of freedom after the penalization, i.e., the cardinality of the estimated active set. $C_T$ is the penalty specific to each criterion, where the most popular choices are: $C_T=2$, the Akaike information criterion (AIC) by \citet{akaike1974new};  $C_T=\log(T)$, the Bayesian information criterion (BIC) by \citet{schwarz1978estimating}.\footnote{Note: for non-Gaussian distributions, the residual sum is often used as a proxy for the likelihood.} The slight modification of the BIC proposed in \citet{wang2009shrinkage} also holds for penalized estimators as the lasso, thus making it consistent asymptotically in both $N$ and $T$.
}
\par With regard to the number of factors and lag-length, as in some applications, the focus is more on forecasting a small subset of time series of the system, we present here two approaches: a \textit{global} approach which gives a single lag-length and number of factors for the entire system and a \textit{local} approach in which the lag-length or number of factors may differ across the time series. We present the two approaches first and then discuss their differences. 

We consider that the factors are driven by a VAR model that is $\BS f_t=\sum_{j=1}^{p_f} \BS \Pi_j \BS f_{t-j}+\BS v_{t-j}$. That means we have two lag-lengths to choose: $p$ and $p_f$.  The one-step ahead forecast error of model \eqref{eq.fac.decom} for the $i$th component is given by \small
$$ \var\Big(x_{i,t}-\sum_{j=1}^{p_f} \BS \Lambda_i^\top \BS \Pi_j \BS f_{t-j}-\sum_{j=1}^p \BS e_i^\top \BS A^{(j)} \BS \xi_{t-j}\Big).$$
\normalsize
If we treat the factors and idiosyncratic components as known, we have to estimate for all components the parameters $\BS \Lambda \in \R^{N\times r}, \BS \Pi_1,\dots,\BS \Pi_{p_f} \in \R^{r\times r}, \BS A^{(1)},\dots,\BS A^{(p)} \in \R^{N\times N}$. Note that $\BS A_1,\dots,\BS A_p$ are sparse. That means in total we have $(N+r p_f)r+ \sum_{j=1}^p \|\BS A_j\|_0$ parameters for all components. For a single component, we treat $\BS \Lambda_i \BS \Pi_j$ as $r$-dimensional vectors which gives in total for the $j$th component $r p_f+ \sum_{j=1}^p \|\be_i^\top\BS A_j\|_0$ parameters. {The sparsity of the idiosyncratic component has the important implication that $\sum_{j=1}^p \|\be_i^\top\BS A_j\|_0$ grow much slower than $Np$. To be precise, the error bounds in Theorem~\ref{thm.est.error.idio} imply that only rates slower than $\sqrt{T}$ are reasonable. Hence, the number of parameters considered grow slower than the sample size and consequently, this fits into the framework of \cite{wang2009shrinkage} and their modified BIC. Note however, that the results of \cite{wang2009shrinkage} are derived under an i.i.d.~set-up and also the pre-selection of the penalty parameter $\lambda_n$ is not taken into account here.} In this modified BIC set-up $C_T$ denotes a slowly diverging series which is discussed shortly. This motivates the following \emph{global} information criteria 
\small
\begin{align}
    \label{crit_FactLags_global}
    IC_{T,N}^{(global)}:=&\underset{r,p,p_f}{\min}\;\log {\frac{1}{NT}\sum_{t=1+\max(p,p_f)}^T \sum_{i=1}^N \left(x_{i,t}
    -\sum_{j=1}^{p_f} \hat{\BS \Lambda}_i^\top \hat{\BS \Pi}_j \hat{\BS f}_{t-j}
    -\sum_{j=1}^p \BS e_i^\top \hat{\bA}^{(j)}\hat{\bxi}_t^{(r)}\right)^2}\\
    &+\bigg(r(p_f+N)+\sum_{j=1}^p\| \hat{\BS A}^{(j)}\|_0\bigg)\frac{\log(T)}{NT} C_T. \nonumber
\end{align}
\normalsize
For the $i$th component we obtain the following \emph{local} information criteria
\small
\begin{align}\label{crit_FactLags}
    IC_{T,N}^{(i)}:=&\underset{r,p,p_f}{\min}\;\log {\frac{1}{T}\sum_{t=1+\max(p,p_f)}^T \left(x_{i,t}
    -\sum_{j=1}^{p_f} \hat{\BS \Lambda}_i^\top\hat{\BS \Pi}_j \hat{\BS f}_{t-j}
    -\sum_{j=1}^p \be_i^\top \hat{\bA}^{(j)}\hat{\bxi}_t^{(r)}\right)^2}\\&+\bigg(rp_f+\sum_{j=1}^p \|\be_i^\top \hat{\BS A}^{(j)}\|_0\bigg)\frac{\log(T)}{T} C_T. \nonumber
\end{align}
\normalsize
In practice, the minimum is evaluated over a finite grid. That means one sets a maximal number of factors $r_{\max}$ and maximal lag-lengths $p_{\max},p_{f,\max}$. If one sets $r_{\max}=0$ or $p_{\max}=0$, this criteria can also be used to fit plain sparse VAR models or plain factor models, respectively. The series $C_T$ can be diverging very slowly and \cite{wang2009shrinkage} suggest for instance, $\log(\log(T))$. We would like to consider the diverging dimension as well and follow a similar route as \citet{bai2002determining}. So we set $C_T=c\frac{\log(NT/(N+T))}{\log(T)}$ with $c=1/2$. Note that for the global approach the factors are penalized by $(p_f+N)\log(NT/(N+T))/(NT)$. This also implies that this series fits into the penalization function framework of Theorem 2 in \cite{bai2002determining} required to obtain a  consistent estimation of the number of factors, i.e., this series converges to $0$ for $N,T\to \infty$ and diverges if scaled by $\min(N,T)$. 


Some remarks on these two information criteria. First, the local approach requires for the $i$th component only an estimation of $\be_i^\top \hat{\bA}^{(j)}$. If the interest is only in some time series of the system, this reduces the computational burden. Second, if the number of factors differs among the time series, the entire system cannot be written as a factor model with a maximal number of factors and a maximal number of lags. Third, the local approach takes into account that large data sets come as a -- in some sense arbitrary -- collection of series and it is most likely that some series are not driven by factors or a small lag-length is sufficient. However, the additional cross-section average in the global approach also leads to more stable results. In simulations, the local approach outperforms the global approach, see Section~\ref{Sec_Simulations} for further discussion.

\section{Numerical Results}\label{Sec_Simulations}

 All results presented in this section are based on implementations in  \emph{R} \citep{R}. We compute the data generating processes (DGPs) at random and consider the following model class:
$\BS x_t=\BS \Lambda \BS f_t+\BS \xi_t, \BS f_t=\sum_{j=1}^{p_f} \BS \Pi^{(j)} \BS f_{t-j}+\BS u_t, \BS \xi_t=\sum_{j=1}^p \BS A^{(j)} \BS \xi_{t-j}+\BS v_t$. The innovations $\{\BS u_t\}, \{\BS v_t\}$ are generated as Gaussian processes and   $\BS \Sigma_u=\var(\BS u_t)$ is generated as a positive definite matrix with eigenvalues in the range $1$ to $10$, using the implementation of the package \emph{clusterGeneration} \citep{Clustergen}. If not denoted otherwise, sparsity of a matrix is obtained by   setting entries -- beginning with the absolute smallest values -- to zero such that the specified amount of sparsity is obtained. The entries of $\bA^{(j)}$ are generated randomly using a $t$ distribution with 3 degrees of freedom. After sparsifying, the matrices are rescaled to fit the eigenvalue conditions of $0.8$. In real data, it is often observed that for a component of a multivariate time series the history of the component itself is quite an important predictor. That means that the diagonals of $\bA^{(j)},j=1,\dots,p$ are (at least for one $j$) often non-sparse. To take this into account we put more weight onto the diagonal of $\bA^{(1)}$ by adding $0.4 \BS I_N$ before sparsifying the randomly generated matrix. This results in a much more dominant diagonal and the diagonal of $\bA^{(1)}$ is more dominant the smaller $p$ is. 

Furthermore, we consider the following specifications: 
\begin{itemize}
    \item The number of factors is given by $r\in\{0,2,4,6\}$.
    \item The sample size is given by $T \in \{100,200\}$.
    \item The dimension is given by $N \in \{50,100,250\}$.
    \item The lag-length of the VAR driving the factors is given by $p_f \in \{0,1,2\}$. The slope matrices are generated at random and the maximal absolute eigenvalue of the stacked VAR matrix is $0.8$. 
    \item The lag-length of the VAR driving the idiosyncratic component is given by $p \in \{0,1,3\}$. The slope matrices are generated at random with a row-wise and column-wise sparsity of $k \in \{5,10,\min(N,100)\}$ and the maximal absolute eigenvalue of the stacked VAR matrix is $0.8$. 
    \item $\BS \Sigma_v=\var(\BS v_t)$ is generated as a positive definite matrix with eigenvalues in the range $1$ to $10$ and sparsity of $k_\Sigma \in \{N/10,N\}$.
    \item The loadings $\Lambda \in \R^{N\times r}$ are generated by random sampling from a Uniform$[-1,1]$ distribution with a column-wise sparsity of $k_\Lambda \in \{N,N/2,N/2^*\}$. $N/2^*$ refers to a setting in which the lower left and upper right part are zero. For this setting, also the the lower left and upper right part of $\BS \Pi^{(j)},j=1,\dots,p_f$ and $\BS \Sigma_u$ are set to zero. 
\end{itemize} 
Note that if the sparsity parameter is of similar size as the dimension, we have no sparsity. Also $r=0$ gives a pure sparse and $p=0$ a pure factor case. Dropping unnecessary combinations, e.g., varying the sparsity for $p=0$, we end up in $2352$ different set-ups for the DGP. We run each set-up $100$ times which results in $235200$ different DGPs. To evaluate the performance, we consider the average one-step ahead prediction error of the first ten time series. To compute the one-step ahead prediction error a test set of $10000$ time points is used. That is, the one-step ahead prediction error of component $i$ is given by $MSFE_{\boldsymbol{{x_i}}}= \Big[ \frac{1}{10000} \sum_{t=1}^{10000} (\hat{{{x}}}_{i,T+t}^{(1)}-{x}_{i,T+t}^{(1)})^2\Big].$ These are then averaged over the $10$ components as well as over the $100$ DGPs of each set-up.

We consider the following models to predict:
\begin{enumerate}[align=parleft]
    \item[\MAR:] Univariate ARs which lag-length is chosen by BIC.
    \item[\MSVAR:] A sparse VAR which is estimated by a row-wise adaptive lasso and the penalty parameter is chosen by BIC. The lag-length is chosen by the local information criteria of Section~\ref{sec_numbfactors} with maximal number of factors equal to zero.
    \item[\MFACABCAR:] A factor model with a VAR for the factors and univariate AR for the idiosyncratic component. The number of factors is chosen by information criteria of \cite{bai2002determining} with the first penalty function, i.e., $g(N,T)=(N+T)/(NT)\log(NT/(N+T))$. The lag-length for the VAR is chosen by BIC and the lag-lengths of the univariate ARs are chosen by AIC.
    \item[\MFACVAR:] The approach presented in this paper, i.e., a factor model with a VAR for the factors and a sparse VAR for the idiosyncratic component. The number of factors and lag-length are chosen by the local information criteria of Section~\ref{sec_numbfactors}. The sparse VAR is estimated by a row-wise adaptive lasso and the penalty parameter is chosen by BIC.
    
\end{enumerate}
{\MFACABCAR} 
use the information criteria of \cite{bai2002determining} to determine the number of factors. Let us mention that in preliminary simulations we also considered the method of \cite{alessi2010improved} as well as the global information criteria of Section~\ref{sec_numbfactors} to determine the number of factors. The obtained results are almost identical to the ones with the information criteria of \cite{bai2002determining}. So for this specification we focus the presentation on the criteria of \cite{bai2002determining} only. Furthermore, for {\MSVAR}, {\MFACVAR} 
we also considered the global information criteria of Section~\ref{sec_numbfactors} but do not present its result here. The findings here can be summarized as follows. The local information criteria outperforms the global criteria and the differences are larger for small sample sizes and dimensions. 
In the following, we present the MSFE-results in relation to the MSFE of {\FLfilt}. That means values larger than $1$ indicate a performance worse than {\FLfilt} and values smaller than $1$ vice versa. The overall performance is summarized in Table~\ref{table.results.all}. 

\begin{table}[H]
\centering
\begin{tabular}{ c| ccccc }
$T$ &\ARBIC & \LAS& \FBRAR \\ \hline
100 & 1.00 & 1.07 & 1.02 \\ 
  200 & 1.13 & 1.09 & 1.13 \\ 
 \end{tabular}
 \caption{Overall performance measured in MSFE and in relation to \FLfilt.} \label{table.results.all}
\end{table}

The relative performance overall $2352$ different DGP set-ups is displayed in Figure \ref{fig:OOS_p_p2_sparsity}-\ref{fig:OOS_p_p2_sparsity_4}. Each dot represents the relative MSFE for one DGP set-up averaged over the runs. The set-ups are sorted by sample size$(T)$, lag-length of the idiosyncratic part$(p)$, sparsity$(k)$, lag-length of the factors$(p_f)$, number of factors$(r)$, and dimension $(N)$. The obtained groups for $T,p,$ and $k$ are highlighted by vertical bars and the specific parameter values are given at the bottom of the figure. This sorting is chosen because these specification parameters matter the most in the sense that the results can differ substantially among different specification of the parameter values.

\begin{figure}[tb]
\centering
\resizebox{\textwidth}{!}{\input{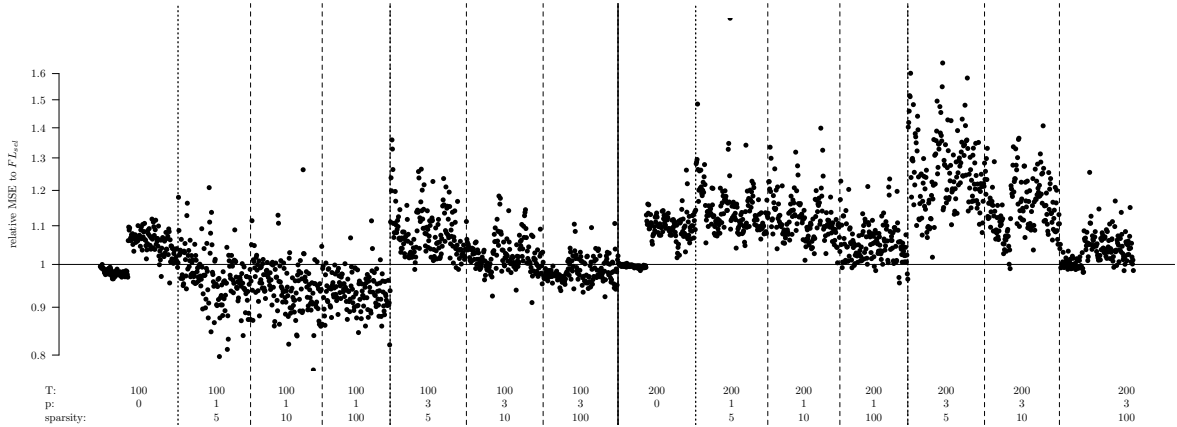}}
\caption{The relative performance of {\ARBIC} over all $2352$ different DGP set-ups. Each dot represents the relative MSFE for one DGP set-up. The set-ups are sorted by sample size,  lag-length of the idiosyncratic part $p$, and sparsity. The obtained groups are highlighted by vertical bars and the specific parameter values are given at the bottom of the figure. Note that a sparsity of $100$ implies in principle here no sparsity at all.}
 \label{fig:OOS_p_p2_sparsity_3}
 \end{figure}

\begin{figure}[tb]
 \centering
     \resizebox{\textwidth}{!}{\input{Simulation_OOS_Results2_3}}
     \caption{The relative performance of {\LAS} over all $2352$ different DGP set-ups. Each dot represents the relative MSFE for one DGP set-up. The set-ups are sorted by sample size,  lag-length of the idiosyncratic part $p$, and sparsity. The obtained groups are highlighted by vertical bars and the specific parameter values are given at the bottom of the figure. Note that a sparsity of $100$ implies in principle here no sparsity at all.}
     \label{fig:OOS_p_p2_sparsity}
\end{figure}

\begin{figure}[tb]
\centering
\resizebox{\textwidth}{!}{\input{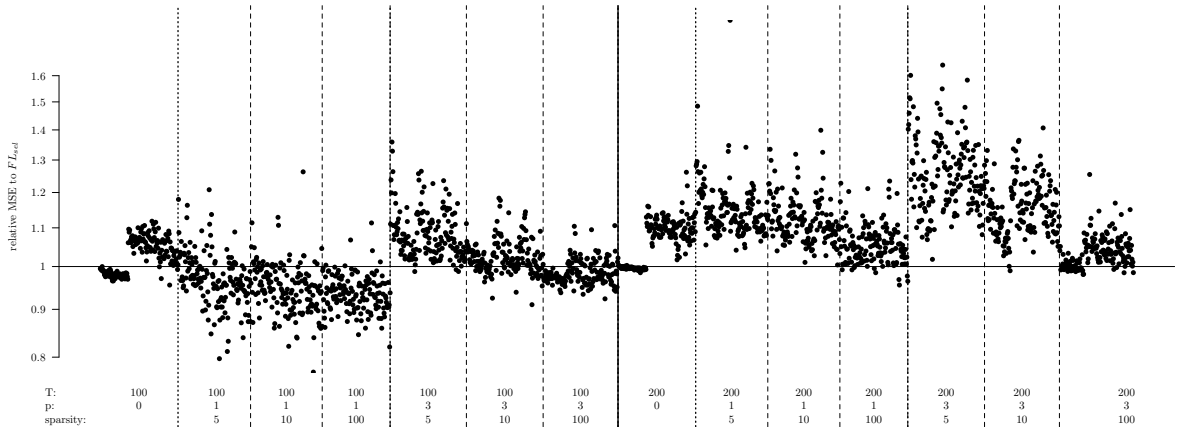}}
\caption{The relative performance of {\FBRAR} over all $2352$ different DGP set-ups. Each dot represents the relative MSE for one DGP set-up. The set-ups are sorted by sample size,  lag-length of the idiosyncratic part $p$, and sparsity. The obtained groups are highlighted by vertical bars and the specific parameter values are given at the bottom of the figure. Note that a sparsity of $100$ implies in principle here no sparsity at all.}
\label{fig:OOS_p_p2_sparsity_4}
\end{figure}

Let us discuss the three Figures~\ref{fig:OOS_p_p2_sparsity_3} to \ref{fig:OOS_p_p2_sparsity_4} starting with the performance relation of {\ARBIC} and {\FLfilt}. If the sample size is small ($T=100$), {\FLfilt} outperform {\ARBIC} only in the case of $p=3$. In all other cases it performs equally good or even worse. 
This behavior changes for the larger sample size settings. Here, {\FLfilt} performs equally well in the cases of no sparsity or no dependence and clearly outperforms {\ARBIC} in the other cases with a smaller MSFE of 40\% or more. \\
In cases in which factors are present, {\FLfilt} is outperforming {\LAS}. The outperformance do not differ much for different sample sizes or lag length and the MSFE is around 10\% smaller for {\FLfilt}.\\
If the sample size is small and the idiosyncratic component is mainly driven by a diagonal VAR (note the construction of slope matrices) {\FBRAR} outperforms {\FLfilt}. It is the other way around for all other cases, i.e., for a less diagonal dominant VAR model and also for larger sample sizes. Then, 
except for the case where the idiosyncratic component has no dependency or no sparsity, {\FLfilt} strongly outperforms {\FBNfilt} with a 10\% to 40\% smaller MSFE. 

To conclude, for the smaller sample size ($T=100$) the additional modelling of the idiosyncratic parts of {\FLfilt} does not always pay off but for the larger sample size ($T=200$) {\FLfilt} does perform best among all competitors and if there is dependence in the idiosyncratic component the gain can be quite substantial. 
 
 \section{Conclusion}\label{sec_conclusions}
  We blend the dense dimensionality reduction of factor models with the one of sparsity-inducing high-dimensional VARs. We propose a \emph{factor model} whose factors and relative loadings are estimated via standard principal components while its idiosyncratic components are assumed to follow a high-dimensional sparse VAR model and are thus estimated via $\ell_1-$norm regularization techniques such as the adaptive lasso. We derive error bounds of this estimation procedure and show in which situations the lasso suffers from the estimation of the idiosyncratic components. 
  {We discuss the implications of our model to forecasting, factor augmented regression, bootstrapping factor models and semi-parametric estimation of the inverse of the spectral density matrix.} 
  To choose the number of factors and the lag-length of the VAR, we propose a unified procedure able to simultaneously estimate both.
  In simulations, we compare the performance of our proposed method with several workhorse forecasting models in the literature and find that the advantage of the procedure proposed can be substantial for moderate to large sample sizes.

\bigbreak

\bigbreak
{\bf Acknowledgments.}  We thank the participants of the workshop   ``Dimensionality Reduction and Inference in High-Dimensional Time Series" at Maastricht University for very helpful comments. The research of the first author was supported by the Research Center (SFB) 884 ``Political Economy of Reforms''(Project B6), funded by the German Research Foundation (DFG). Furthermore, the first author acknowledges support by the state of Baden-W{\"u}rttemberg through bwHPC.

\bibliography{literature.bib}
\newpage
\begin{appendices}
\numberwithin{table}{section}
\numberwithin{equation}{section}
\numberwithin{lemma}{section}
\section{Proofs and additional Lemmas}
In order to quantify the dependence of the stochastic processes, we use the concept of functional dependence, see \cite{wu2005nonlinear}, and concentration inequalities derived under this concept of dependence, see among others \cite{liu2013probability,wu2016performance}. In the following remark~\ref{remark.functional.dependence} we summarize the main notation of this dependence concept.

\begin{remark}[Functional Dependence Measure]\label{remark.functional.dependence}
Let $Y_{t;i}=G_i(\eps_t,\eps_{t-1},\dots,), i=1,\dots,N, t \in \Z,$ be some process generated causally by the i.i.d. processes $\{\eps_t\}$ for some measurable function $G=(G_1,\dots,G_N)$. Furthermore, denote by $Y_{t;i}^{\prime(k)}=G_i(\eps_t,\eps_{t-1},\dots,\eps_{t-k+1},\eps_{t-k}^\prime,\eps_{t-k-1},\eps_{t-k-2},\dots)$ the process where $\eps_{t-k}$ is replaced by an i.i.d.~copy $\eps_{t-k}^\prime$. We follow \cite{wu2005nonlinear,wu2016performance} and  define the physical/functional dependence coefficients in the following way. Let $\left\|\xi_{i,t}\right\|_{E,q}:=\left(\mathbb{E}\left|\xi_{i,t}\right|^{q}\right)^{1 / q}<\infty, q \geq 1$. Furthermore, let the functional dependence measure be defined as $\delta_{k,q,i}=\|Y_{0;i}-Y_{0;i}^{\prime(k)}\|_{E,q}, k \geq 0$. In order to account for the dependence in the process $Y_{;i}$ let $\Delta_{m,q;i}=\sum_{k=m}^\infty \delta_{k,q;i}$ such that the dependence adjusted norm is defined as $\|Y_{;i}\|_{q,\alpha}=\sup_{m\geq 0} (m+1)^\alpha \Delta_{m,q;i}$. As we work in a high-dimensional setting, in order to take this into account we need a uniform dependence adjusted norm $\Psi_{q,\alpha}=\max_{1\leq i\leq N} 
\|Y_{\cdot;i}\|_{q,\alpha},$ and an overall dependence adjusted norm $\Upsilon_{q,\alpha}=(\sum_{i=1}^N (\sup_{m\geq 0} (m+1)^\alpha \Delta_{m,q;i})^q)^{1/q}$. Furthermore, define for the $N$ dimensional stationary process $Y_{t;i}$ the $\mathcal{L}^{\infty}$ functional dependence measure with its corresponding dependence adjusted norm: 
$\omega_{k,q}= \| \|Y_{t;i}-Y_{t;i}^{\prime(k)}\|_{\infty}\|_{E,q}$,
$\| \|Y_\cdot\|_\infty \|_{q,\alpha}=\sup_{m\geq 0} (m+1)^\alpha \Omega_{k,q}$ for $\Omega_{m,q}=\sum_{k=m}^\infty \omega_{k,q}$.
Finally, let $\nu_q=\sum_{j=1}^\infty (j^{q/2-1} \omega_{k,q})^{1/(q+1)}$.

Assumption~\ref{sparsity} implies that $\|e_i^\top \BS B^{(j)}\|_2\leq M\rho^j$. Hence, it follows by Example 3 in \cite{wu2016performance} and the moment condition in Assumption~\ref{ass.moments} that $\max_j \|\{\xi_{j,t}\}\|_{\zeta,\alpha}<\infty$ for all $\alpha>0$. Since $\{\BS f_t\}$ is a linear processes of fixed dimension $r$, we also have $\max_j \|\{ \chi_{j,t}\}\|_{\zeta,\alpha}<\infty$ for all $\alpha>0$. Hence, we have by the Minkowski-inequality $\max_j \|\{x_{j,t}\}\|_{\zeta,\alpha}<\infty$, see also the Proof of Proposition~5 in \cite{forni2017dynamic}. Additionally, we have by the Cauchy-Schwarz-inequality for some $q>2$, $\max_{j,i} \|\{\xi_{j,t} f_{i,t}\}\|_{q,\alpha}\leq C(\max_{j} \|\{\xi_{j,t}\}\|_{2q,\alpha}+
\max_{j} \|\{f_{i,t}\}\|_{2q,\alpha}+\max_{j} \|\{\xi_{j,t}\}\|_{2q,\alpha}\max_{j} \|\{f_{i,t}\}\|_{2q,\alpha})$, where $C$ is some constant depending on $q$ only.
\end{remark}

In the following lemma, we derive the  order of several expression. Key ingredient of the proof of this lemma is the Nagaev's inequality for dependent processes, see Section 2.1 in \cite{wu2016performance}. To abbreviate the expression, we display all results here in $O_p$-notation. The proof of all lemmas presented here can be found in the supplementary material. Also, recall $$\gt=(NT)^{2/\zeta}\left(\frac{1}{\sqrt{N}T}+\frac{1}{T^{3/2}}+(NT)^{2/\zeta}\frac{1}{T^2}\right).$$
\begin{lemma} \label{lemma.1}
Let $C_1,C_2,C_3$ be constants depending only on $q$ and $\alpha$. Under Assumption~\ref{sparsity},\ref{ass.moments},\ref{ass.fac} we  have the following:
\begin{enumerate}[A)]
\item \label{lem.1.part.0} $\|\BS D_{NT,r}^2\|_2=O_P(1)$ and $\|\BS D_{NT,r}^{-2}\|_2=O_P(1).$
\item \label{lem.1.part.A} For $i=1,\dots,N$ and $j=1,\dots,r,$ we have
$$
\max_{i,j} |1/T\sum_{s=1}^T \xi_{i,s} f_{j,s}|=O_P(\sqrt{(\log(N)/T})+N^{2/\zeta}T^{2/\zeta-1}).
$$

\item \label{lem.1.part.B} For each $j=1,\dots,r,$ we have
 $1/T\sum_{s=1}^T | f_{s,j}|^2=\Sigma_{F,j,j}+O_P(1/\sqrt{T})$ and  $\max_j |1/T\sum_{s=1}^T f_{s,j}\hat { f}_{s,l}|\leq (M+O_P(1/\sqrt{T}))^{1/2}$.

\item \label{lem.1.part.C} For each $j_1,j_2=1,\dots,N$, we have 
$$\max_{j_1,j_2} |1/T\sum_{s=1}^T \xi_{s,j_1} \xi_{s,j_2}|\leq M+O_P(\sqrt{(\log(N)/T})+N^{2/\zeta}T^{2/\zeta-1}).$$
\item \label{lem.1.part.D} For each $k=1,\dots,r$ we have $e_k^\top \BS\Lambda^\top \Gamma_{\xi}(0) \BS \Lambda e_k/N\leq M^4/(1-\rho^2)<\infty$ and 
$$1/T\sum_{s=1}^T (1/\sqrt{N} \sum_{i=1}^N \ell_{i,k} \xi_{i,s})^2=O_P(1).$$
\item \label{lem.1.part.E} 
We have 
$$\max_{j,k} |1/T\sum_{s=1}^T \be_j^\top \BS \xi_t \BS \xi_t^\top \BS \Lambda^\top \be_k|=O_P(k_\xi+\sqrt{\log(N)/T}+N^{2/\zeta}T^{2/\zeta-1}).$$
\item \label{lem.1.part.F} We have 
\begin{align*}
    &\max_{j,l} \left|1/T\sum_{s=1}^T \be_l^\top(\hat{\BS f}_s - \bH_{NT} {\BS f}_s) \xi_{i,s}\right|=O_P\left(\frac{{\log(N)}}{{T}}+\frac{k_\xi}{N}+\frac{\sqrt{\log(N)}}{\sqrt{NT}}+\gt\right).
\end{align*}\normalsize
\item \label{lem.1.part.G}
We have for $t\in \Z$
\begin{align*}
    \hat{\boldsymbol{f_t}}-\BS H_{NT}\boldsymbol{f_t}=&\frac{1}{NT}\left[\sum_{i=1}^N \sum_{s=1}^T \xi_{i,t} \BS\Lambda_i \BS f_{s}^\top \BS H_{NT}{\BS f_s}+ \sum_{i=1}^N \sum_{s=1}^T \xi_{i,t}\xi_{i,s} \BS H_{NT}{\BS f_s}\right] \BS D_{NT,r}^{-2}\\
    &+O_P\left(\frac{{\log(N)}}{{T}}+\frac{k_\xi}{N}+\frac{\sqrt{\log(N)}}{\sqrt{NT}}+\gt\right). 
\end{align*}\normalsize
\item \label{lem.1.part.H} We have\small
\begin{align*}
    &\max_{j,l}\left|\frac{1}{T}\sum_{s=1}^T f_{j,s} [\hat{\BS f}_s - \BS H_{NT} \BS f_s]^\top \be_l\right|=O_P\left(\frac{{\log(N)}}{{T}}+\frac{k_\xi}{N}+\frac{\sqrt{\log(N)}}{\sqrt{NT}}+\gt\right).
\end{align*}\normalsize
\item \label{lem.1.part.I} 

For each $j,l=1,\dots,r$
\small
\begin{align*}
    &\frac{1}{T}\sum_{s=1}^T \be_j^\top[\hat {\BS f}_s - H_{NT} {\BS f}_s] [\hat{\BS f}_s - \BS H_{NT} \BS f_s]^\top \be_l=O_P\left(\frac{{\log(N)}}{{T}}+\frac{k_\xi}{N}+\frac{\sqrt{\log(N)}}{\sqrt{NT}}+\gt\right).
\end{align*}\normalsize
\item \label{lem.1.part.J} We have\small
\begin{align*}
    (\BS H_{NT}^\top)^{-1} \BS \Lambda_i - \hat{\BS \Lambda}_i=\frac{1}{T} \sum_{s=1}^T \bH_{NT} \BS f_s \xi_{i,s}+Error_i,
    \end{align*}\normalsize
    where $\max_i |Error_i|=O_P\left(\frac{{\log(N)}}{{T}}+\frac{k_\xi}{N}+\frac{\sqrt{\log(N)}}{\sqrt{NT}}+\gt\right)$. Furthermore,
    $$\|(\BS H_{NT}^\top)^{-1} \BS \Lambda - \hat{\BS \Lambda}\|_{\max}=O_P(\sqrt{\log(N)/T}+(NT)^{2/\zeta}/T+\frac{k_\xi}{N}).$$

\end{enumerate}
\end{lemma}

\bigbreak
\begin{proof}[Proof of Theorem~\ref{thm.est.error.step.one}]
First note that we have the following representation
\begin{align*}w_{i,t}&=\BS \Lambda_i^\top \BS f_t- \hat{\BS \Lambda}_i^\top \hat{\BS f}_t=\BS \Lambda_i^\top \BS H_{NT}^{-1} [\BS H_{NT} \BS f_t - \hat{\BS f}_t]  + [(\BS H_{NT}^\top)^{-1} \BS \Lambda_i - \hat{\BS \Lambda}_i]^\top \BS H_{NT} \BS f_t  \\&+ [(\BS H_{NT}^\top)^{-1} \BS \Lambda_i - \hat{\BS \Lambda}_i]^\top[\BS H_{NT} \BS f_t - \hat{\BS f}_t],\end{align*}
Then, the first assertion in \eqref{corrupt_repres} follows by inserting the orders derived in Lemma~\ref{lemma.1}. 

For the second and third assertion of Theorem ~\ref{thm.est.error.step.one}, namely for the orders of $\|{1}/{T} \sum_{t=1}^T \BS w_t \BS \xi_t^\top \|_{\max}$ and $ \|{1}/{T} \sum_{t=1}^T \BS w_t \BS w_t^\top\|_{\max}$ respectively, we focus on the case $k=0$. The case $k=1$ follows by the same arguments. For the second assertion, we have by Lemma~\ref{lemma.1}
\small
\begin{equation*}
\begin{aligned}
        &\Big\|\frac{1}{T} \sum_{t=1}^T \BS w_t \BS \xi_t^\top \Big\|_{\max} \leq \|\BS \Lambda\|_{\max} \|\BS H_{NT}\|_{\max}^2 \|\BS D_{NT,r}^{-2}\|_{\max}\Bigg[
        \Big\| \frac{1}{NT} \sum_{t=1}^T \BS \Lambda^\top \BS \xi_t \BS \xi_t^\top \Big\|_{\max} \Big\|\frac{1}{T} \sum_{t=1}^T \BS f_t^\top \BS f_t\Big\|_{\max}+\\
        &+\max_i |Error_i|+\Big\|\frac{1}{T} \sum_{s=1}^T \BS f_s^\top \BS \xi_s\Big\|_{\max} \left(\Big\|\frac{1}{T} \sum_{t=1}^T \BS \xi_t \BS \xi_t^\top-\BS \Gamma_\xi(0)\Big\|_{\max}+\|\BS \Gamma_\xi(0)\|_\infty/N\right)\Bigg]+\Big\|\frac{1}{T} \sum_{t=1}^T \BS f_t^\top \BS \xi_t\Big\|_{\max}^2\\
        =&O_P\left(\frac{k_\xi}{N}+\frac{\sqrt{\log(N)}}{N\sqrt{T}}+(NT)^{2/\zeta-1}+\frac{\log(N)}{T}+(NT)^{4/\zeta}/T^2+\frac{k_\xi}{N}\left(\sqrt{(\log(N)/T})+N^{2/\zeta}T^{2/\zeta-1}\right)\right)\\
        &+O_P\left(\frac{{\log(N)}}{{T}}+\frac{k_\xi}{N}+\frac{\sqrt{\log(N)}}{\sqrt{NT}}+\gt\right).
    \end{aligned}
\end{equation*}\normalsize
For the third assertion, we have by Lemma~\ref{lemma.1}\small
\begin{align*}
        \Big\|\frac{1}{T} \sum_{t=1}^T& \BS w_t \BS w_t^\top\Big\|_{\max}\leq \|\BS \Lambda\|_{\max}^2 \|\BS H_{NT}\|_{\max}^4 \|\BS D_{NT,r}^{-2}\|_{\max}^2 \Bigg[ \Big\|\frac{1}{T} \sum_{s=1}^T \BS f_s^\top \BS f_s \Big\|^2_{\max} \Big\|\frac{1}{NT} \sum_{t=1}^T \BS \Lambda^\top \BS \xi_t\BS \xi_t^\top \BS \Lambda \Big\|_{\max}/N\\
        &+\Big\|\frac{1}{T} \sum_{s=1}^T \BS f_s^\top \BS f_s \Big\|_{\max} \Big\|\frac{1}{T} \sum_{s=1}^T \BS f_s^\top \BS \xi_s\Big\|_{\max}\left(2\Big\|\frac{1}{NT} \sum_{t=1}^T\BS \Lambda^\top \BS \xi_t \BS \xi_t\Big\|_{\max}+\Big\|\frac{1}{NT}\sum_{t=1}^T \BS \Lambda^\top \BS \xi_t \BS f_t\Big\|_{\max}\right) \\
        &+\Big\|\frac{1}{T} \sum_{s=1}^T \BS f_s^\top \BS \xi_s\Big\|_{\max}^2\left(\Big\|\frac{1}{T}\sum_{t=1}^T \BS \xi_t \BS \xi_t^\top -\BS \Gamma_\xi(0)\Big\|_{\max}+\|\BS \Gamma_\xi(0)\|_\infty/N\right)+\Big\|\frac{1}{T} \sum_{s=1}^T \BS f_s^\top \BS \xi_s\Big\|_{\max}^3\Bigg]\\
        &+\Big\|\frac{1}{T} \sum_{s=1}^T \BS f_s^\top \BS \xi_s\Big\|_{\max}^2\Big\|\frac{1}{T} \sum_{s=1}^T \BS f_s^\top \BS f_s \Big\|_{\max}+\max_i|Error_i|\\
        =&O_P\left(\frac{1}{N}+\left(\frac{\sqrt{\log(N)}}{\sqrt{T}}+(NT)^{2/\zeta}/T\right)\left(\frac{k_\xi}{N}+\frac{\sqrt{\log(N)}}{\sqrt{T}}+(NT)^{2/\zeta}/T\right)\right)\\
        &+O_P\left(\frac{{\log(N)}}{{T}}+\frac{k_\xi}{N}+\frac{\sqrt{\log(N)}}{\sqrt{NT}}+\gt\right)\\
        =&O_P\left(\frac{k_\xi}{N}+\frac{\log(N)}{T}+\frac{\sqrt{\log(N)}}{\sqrt{NT}}+\gt\right).
        \end{align*}\normalsize
\end{proof}


\begin{lemma} \label{lem_spars}
Under Assumption~\ref{sparsity} and if \small$$\left\|\frac{1}{T-p} \sum_{t=p+1}^T  (\hat \xi_{j,t}- \bbeta^\top \hat \bxi_{t-1}^v) \hat \bxi_{t-1}^v\right\|_{\max} \leq \hat \lambda/4, $$\normalsize and \small    $$\Theta^\top \frac{1}{T-p} \sum_{t=p+1}^T \hat \bxi_{t-1}^v (\hat \bxi_{t-1,j})^\top \Theta\geq \alpha \|\Theta\|_2^2-\hat\tau \|\Theta\|_1^2 \, \forall\, \Theta \in \R^{np},$$ \normalsize we have \small
$$
\| \hat{\bbeta}^{(j)}-{\bbeta}^{(j)}\|_2 \leq 16 \max\left(\sqrt{k} (\hat \lambda/\alpha)^{1-q/2},\sqrt{\hat \tau} s (\hat \lambda/\alpha\right)^{1-q},
$$\normalsize
and \small
$$
\|\hat{\bbeta}^{(j)}-{\bbeta}^{(j)}\|_1 \leq \max(68 k (\hat \lambda/\alpha)^{1-q},64 \sqrt{\tau}k^{3/2} (\hat \lambda/\alpha)^{1-3/2q}+4k(\hat \lambda/\alpha)^{1-q}).
$$\normalsize
\end{lemma}
\bigbreak
\begin{proof}[Proof of Theorem~\ref{thm.est.error.idio}]
The idea is to determine the order of the quantities $\hat \lambda_T$ and $\hat \tau$ in Lemma~\ref{lem_spars}. For this, first note that since $(\xi_{j,t}- \BS \beta_j^\top \BS \xi_{t-1}^v)\BS\xi_{t-1}^v=v_{j,t} \BS \xi_{t-1}$ and $\mathbb{E}v_{j,t} \BS \xi_{t-1}^v=0$, we have $\|\frac{1}{T-p} \sum_{t=p+1}^T (\BS\xi_{j,t}- \BS \beta_j^\top \BS \xi_{t-1}^v)\BS \xi_{t-1}^v\|_{\max}=O_P(\sqrt{\log(Np)/T}+(NpT)^{2/\zeta}/T)$ by the same arguments as in the proof of Lemma~\ref{lemma.1} \ref{lem.1.part.D} and note that $\BS\xi_{t-1}^v$ is of dimension $Np$. Additionally, we have $\max_j\|\BS \beta_j\|_1\leq M^{1-q} k$ and \small \begin{equation*}\begin{aligned}\hat\lambda_T&=\left\|\frac{1}{T-p} \sum_{t=p+1}^T (\hat \xi_{j,t}- \bbeta^\top \hat \xi_{t-1}^v)\hat {\BS \xi}_{t-1}^v\right\|_{\max}\leq  \left\|\frac{1}{T-p} \sum_{t=p+1}^T (v_{j,t})\BS \xi_{t-1}^v\right\|_{\max}+\left\|\frac{1}{T-p} \sum_{t=p+1}^T v_{j,t} \BS w_{t-1}^v\right\|_{\max}\\ &\left\|\frac{1}{T-p} \sum_{t=p+1}^T  w_{j,t}\BS \xi_{t-1}^v+w_{j,t}\BS w_{t-1}^v\right\|_{\max}+
\|\bbeta\|_1\left\|\frac{1}{T-p} \sum_{t=p+1}^T\BS w_{t-1}^v (\BS \xi_{t-1}^v)^\top+\BS w_{t-1}^v (\BS w_{t-1}^v)^\top\right\|_{\max}\\
&=O_P\Bigg(\sqrt{\log(Np)/T}+(NpT)^{2/\zeta}/T+k\Bigg(\frac{k_\xi}{N}+\frac{\log(Np)}{T}+\frac{\sqrt{\log(Np)}}{\sqrt{NT}}+\gt\Bigg).
\end{aligned}\end{equation*}\normalsize

Let $\Gamma=\var((\BS \xi_t^\top,\dots,\BS\xi_{t-p+1}^\top)^\top)$ and $\hat \Gamma=\frac{1}{T-p} \sum_{t=p}^{T-1} \hat{\BS \xi}_{t}^v (\hat{\BS\xi}_{t}^v)^\top$. We have for $\Theta\in \R^{Np}$ $\Theta^\top \hat \Gamma \Theta=\Theta^\top \Gamma \Theta+\Theta^\top(\hat \Gamma-\Gamma)\Theta\geq \alpha \|\Theta\|_2^2  - \|\Theta\|_1^2\|\hat \Gamma-\Gamma\|_{\max}\geq\alpha \|\Theta\|_2^2  - \|\Theta\|_1^2\|(\|\Gamma-\frac{1}{T-p} \sum_{t=p}^{T-1} {\BS \xi}_{t}^v ({\BS\xi}_{t}^v)^\top\|_{\max}+2\|\frac{1}{T-p} \sum_{t=p}^{T-1} {\BS \xi}_{t}^v ({\BS w}_{t}^v)^\top\|_{\max}+\{\frac{1}{T-p} \sum_{t=p}^{T-1} \hat{\BS w}_{t}^v ({\BS w}_{t}^v)^\top\|_{\max})=:\alpha\|\Theta\|_2^2+\hat \tau\|\Theta\|_1^2$. With the results of Theorem~\ref{thm.est.error.step.one} and Lemma~\ref{lemma.1}, we have  $\hat \tau=O_P(\sqrt{\log(Np)/T}+(NpT)^{2/\zeta}/T+k\Big(\frac{k_\xi}{N}+\frac{\log(Np)}{T}+\frac{\sqrt{\log(Np)}}{\sqrt{NT}}+(NpT)^{2/\zeta}(\frac{k_\xi}{NT}+\frac{1}{\sqrt{N}T}+\frac{1}{T^{3/2}}+(NpT)^{2/\zeta}\frac{1}{T^2})\Big)$. That means $\hat \tau=O_p(\hat \lambda_T)$. By plugging this into Lemma~\ref{lem_spars} we then obtain
$
\|\hat \beta_j-\beta_j\|_2=O_P(\sqrt{k} (\hat \lambda_T)^{1-q/2}+k(\hat \lambda_T)^{3/2-q/2})
$
and \\
$
\|\hat \beta_j-\beta_j\|_1 =O_P(k (\hat \lambda)^{1-q}+k^{3/2}(\hat \lambda_T)^{3/2(1-q)})=O_P(k (\hat \lambda)^{1-q}).
$

\end{proof}

\begin{proof}[Proof of Theorem~\ref{thm.prediction}]
First note that $\BS e_j^\top(\hat {\BS X}_{T+1}^{1,p_2}-\BS X_{T+1}^{(1,p_2)})=\hat {\BS \beta}_j\hat {\BS{\xi}}_T^v - \BS \beta_j {\BS{\xi}}_T^v+\hat {\BS \Lambda}_j^\top \sum_{i=1}^{p_2} \hat {\BS \Pi}_i^{(p_2)} \hat {\BS f}_{T+1-i}- \BS \Lambda_j^\top \BS H_{NT}^{-1} \BS H_{NT} \sum_{i=1}^{p_2} {\BS \Pi}_i^{(p_2)} \BS H_{NT}^{-1} \BS H_{NT} {\BS f}_{T+1-i}$. Then, the results derived in  Theorem~\ref{thm.est.error.step.one}, \ref{thm.est.error.idio} and Lemma~\ref{lemma.1} can be plugged in. Note further that due to due to Assumption~\ref{ass.moments} and \ref{ass.fac} we have\\ $\BS \Lambda \BS \xi_t/N=\|\BS \Lambda\|_2/N \BS \Lambda/\|\BS \Lambda\|_2 \BS \xi_t=O_P(1/\sqrt{N})$ appearing in $\hat{\boldsymbol{f_t}}-\BS H_{NT}\boldsymbol{f_t}$  and the assertion follows.
\end{proof}

\begin{lemma}\label{lemma.spectral.density.factor}
Under Assumption~\ref{ass.moments}, \ref{ass.fac}, \ref{sparsity.b} and Assumption~1 and 2 in \cite{wu2018asymptotic} we  have the following \small
\begin{align*}
    \| \boldsymbol{\hat{f}}_f(\omega)-\BS H_{NT}\BS f_f(\omega)\BS H_{NT}^\top\|_{\max}=O_P\Bigg(&\sqrt{B_T/T}+\frac{1}{B_T}\Bigg(\log(N)+\frac{k_\xi T}{N}+\frac{\sqrt{\log(N)T}}{\sqrt{N}}\\
    &+(NT)^{2/\zeta}T\left(\frac{1}{\sqrt{N}T}+\frac{1}{T^{3/2}}+(NT)^{2/\zeta}\frac{1}{T^2}\right)\Bigg)\Bigg),
\end{align*}

\begin{align*}
    \|\boldsymbol{\hat f}_f^{-1}(\omega)-(\BS H_{NT}\boldsymbol{f}_f(\omega)\BS H_{NT}^\top)^{-1}\|_{l}=O_P\Bigg(&\sqrt{B_T/T}+\frac{1}{B_T}\Bigg(\log(N)+\frac{k_\xi T}{N}+\frac{\sqrt{\log(N)T}}{\sqrt{N}}\\
    &+(NT)^{2/\zeta}T\left(\frac{1}{\sqrt{N}T}+\frac{1}{T^{3/2}}+(NT)^{2/\zeta}\frac{1}{T^2}\right)\Bigg)\Bigg), l \in [1,\infty].
\end{align*}
\end{lemma}\normalsize

\begin{lemma} \label{lem.rate.spectral.var}
Under Assumption~\ref{ass.moments}, \ref{ass.fac}, \ref{sparsity.b} we  have the following:\small
\begin{align*}
    \| &\hat {\BS \Sigma}_v^{-1,CLIME}-{\BS \Sigma}_v^{-1}\|_l=O_P\Bigg(k_v \|{\BS \Sigma}_v^{-1}\|_1 \Bigg[\sqrt{(\log(N)/T})+N^{2/\zeta}T^{2/\zeta-1}+k\Big[\frac{k_\xi}{N}+\frac{\log(N)}{T}+\frac{\sqrt{\log(N)}}{\sqrt{NT}}\\
    &+(NT)^{2/\zeta-1}k_\xi+\frac{(NT)^{4/\zeta}}{T^2}\Big]+(\sqrt{(\log(Np)/T})+(Np)^{2/\zeta}T^{2/\zeta-1})\Bigg(k\Bigg[\sqrt{\log(Np)/T}+(NpT)^{2/\zeta}/T\\
    &+k\Bigg(\frac{k_\xi}{N}+\frac{\sqrt{\log(Np)}}{\sqrt{NT}}+(NpT)^{2/\zeta}\Bigg(\frac{k_\xi}{NT}+\frac{1}{\sqrt{N}T}+\frac{1}{T^{3/2}}+(NpT)^{2/\zeta}\frac{1}{T^2}\Bigg)\Bigg)\Bigg]^{1-q}\Bigg)\Bigg]^{1-q_v}\Bigg), l \in [1,\infty],
\end{align*}

\begin{align*}
    \|\BS f_\xi(\omega)^{-1}-\hat{\BS f}_\xi(\omega)^{-1}\|_l=O_P( k^2 (\|{\BS \Sigma}_v^{-1}-\hat {\BS \Sigma}_v^{-1,CLIME}\|_l+\max_s  \|\hat \bbeta^{(s)}- \bbeta^{(s)}\|_2^{1-q} \|{\BS \Sigma}_v^{-1}\|_l ), l \in [1,\infty],
\end{align*}
\begin{align*}
    \|\BS f_\xi(\omega)^{-1}-\hat{\BS f}_\xi(\omega)^{-1}\|_2=O_P(\|{\BS \Sigma}_v^{-1}-\hat {\BS \Sigma}_v^{-1,CLIME}\|_l+k\max_s  \|\hat \bbeta^{(s)}- \bbeta^{(s)}\|_2^{1-q} ).
\end{align*}
\normalsize
If $N=T^a, p=T^b$ for some $a,b>0$, $\zeta\geq 4(1+a+b)$ and $k$ such that $\|\hat \A-\A\|_\infty=o_P(1)$, these error bounds simplify to 
\small
\begin{align*}
    \| &\hat {\BS \Sigma}_v^{-1,CLIME}-{\BS \Sigma}_v^{-1}\|_l=O_P\Bigg(k_v \|{\BS \Sigma}_v^{-1}\|_1 \Bigg[\sqrt{(\log(N)/T})+k\Big[\frac{k_\xi}{N}+\frac{\log(N)}{T}+\frac{\sqrt{\log(N)}}{\sqrt{NT}}\Big]\Bigg]^{1-q_v}\Bigg), l \in [1,\infty],
\end{align*}

\begin{align*}
    \|\BS f_\xi(\omega)^{-1}-\hat{\BS f}_\xi(\omega)^{-1}\|_l=O_P\Bigg(& k^2 \|{\BS \Sigma}_v^{-1}\|_1 \Big(k_v  \Big[\sqrt{(\log(N)/T})+k\Big[\frac{k_\xi}{N}+\frac{\log(N)}{T}+\frac{\sqrt{\log(N)}}{\sqrt{NT}}\Big]\Big]^{1-q_v}+\\
    & \sqrt{k} [\sqrt{\log(Np)/T}+k k_\xi/N+k \sqrt{\log(Np)/(NT)}]^{1-q/2} \Big)\Bigg), l \in [1,\infty],
\end{align*}
\begin{align*}
    \|\BS f_\xi(\omega)^{-1}-\hat{\BS f}_\xi(\omega)^{-1}\|_2&=O_P\Big(k_v \|{\BS \Sigma}_v^{-1}\|_1 \Bigg[\sqrt{(\log(N)/T})+k\Big[\frac{k_\xi}{N}+\frac{\log(N)}{T}+\frac{\sqrt{\log(N)}}{\sqrt{NT}}\Big]\Bigg]^{1-q_v}\\&+k^{3/2} [\sqrt{\log(Np)/T}+k k_\xi/N+k \sqrt{\log(Np)/(NT)}]^{1-q/2} \Big).
\end{align*}
\end{lemma}
\normalsize
\bigbreak
\begin{proof}[Proof of Theorem~\ref{thm.spectral.density}]
We have \small
\begin{align}
    &\|\boldsymbol{f}_X(\omega)^{-1}-\boldsymbol{\hat f}_X(\omega)^{-1}\|_l\leq \label{eq.spec.DFM}\\
&\|\BS f_\xi^{-1}(\omega)-\hat{\BS f}_\xi^{-1}(\omega)\|_l+\|\hat{\BS f}_\xi^{-1}(\omega) \hat \Lambda \Big(\boldsymbol{\hat f}_f^{-1}(\omega)/N+\boldsymbol{\hat \Lambda}^\top/\sqrt{N} \hat{\BS f}_\xi^{-1}(\omega)\boldsymbol{ \hat \Lambda}/\sqrt{N}\Big)^{-1} \boldsymbol{\hat \Lambda}^\top/N \hat{\BS f}_\xi^{-1}(\omega) \nonumber\\-&\BS f_\xi^{-1}(\omega) \BS \Lambda \BS H_{NT}^{-1} \Big((\BS H_{NT}^{-1})^\top \boldsymbol{f}_f^{-1}(\omega)\BS H_{NT}^{-1}/N+(\BS H_{NT}^{-1})^\top\BS \Lambda^\top /\sqrt{N}\BS f_\xi^{-1}(\omega) \BS \Lambda \BS H_{NT}^{-1}/\sqrt{N}\Big)^{-1} (\BS H_{NT}^{-1})^\top \BS \Lambda^\top/N \BS f_\xi^{-1}(\omega)
\|_l \nonumber, 
\end{align}\normalsize
Let \small$G=((\BS H_{NT}^{-1})^\top  \ f_f^{-1}(\omega)\BS H_{NT}^{-1}/N+(\BS H_{NT}^{-1})^\top\BS \Lambda^\top \BS f_\xi^{-1}(\omega) \BS \Lambda \BS H_{NT}^{-1}/N)^{-1}$\normalsize and \small$\hat G=(\boldsymbol{\hat f}_f^{-1}(\omega)/N+\boldsymbol{\hat  \Lambda}^\top \hat{\BS f}_\xi^{-1}(\omega) \boldsymbol{\hat \Lambda})^{-1}/N$. \normalsize Lemma~\ref{lem.rate.spectral.var} gives a rate for $\|\BS f_\xi^{-1}(\omega)-\hat{\BS f}_\xi^{-1}(\omega)\|_l$. Furthermore, the second term on the right hand side of \eqref{eq.spec.DFM} is smaller or equal to:\small
\begin{align*}
&\|\hat{\BS f}_\xi^{-1}(\omega) \boldsymbol{\hat \Lambda}-{\BS f}_\xi^{-1}(\omega) \BS \Lambda (\BS H_{NT}^{-1})\|_l \| G\|_l \|(\BS H_{NT}^{-1})^\top \BS \Lambda^\top/N \BS f_\xi^{-1}(\omega)\|_l + \| \BS f_\xi^{-1}(\omega) \BS \Lambda \BS H_{NT}^{-1}\|_l  \|G-\hat G\|_l \|(\BS H_{NT}^{-1})^\top \BS \Lambda^\top/N \BS f_\xi^{-1}(\omega)\|_l \\
&+ \| \BS f_\xi^{-1}(\omega) \BS \Lambda \BS H_{NT}^{-1}\|_l  \|G\|_l \| \boldsymbol{\hat \Lambda}^\top/N \hat{\BS f}_\xi^{-1}(\omega)-(\BS H_{NT}^{-1})^\top \BS \Lambda^\top/N \BS f_\xi^{-1}(\omega)\|_l \\
&+
\|\hat{\BS f}_\xi^{-1}(\omega) \boldsymbol{\hat \Lambda}-{\BS f}_\xi^{-1}(\omega) \BS \Lambda (\BS H_{NT}^{-1})\|_l \|G-\hat G\|_l \|(\BS H_{NT}^{-1})^\top \BS \Lambda^\top/N \BS f_\xi^{-1}(\omega)\|_l\\
&+\| \BS f_\xi^{-1}(\omega) \BS \Lambda \BS H_{NT}^{-1}\|_l \|G-\hat G\|_l \| \boldsymbol{\hat \Lambda}^\top/N \hat{\BS f}_\xi^{-1}(\omega)-(\BS H_{NT}^{-1})^\top \BS \Lambda^\top/N \BS f_\xi^{-1}(\omega)\|_l\\
&+\|\hat{\BS f}_\xi^{-1}(\omega) \boldsymbol{\hat \Lambda}-{\BS f}_\xi^{-1}(\omega) \BS \Lambda (\BS H_{NT}^{-1})\|_l \| G\|_l \| \boldsymbol{\hat \Lambda}^\top/N \hat{\BS f}_\xi^{-1}(\omega)-(\BS H_{NT}^{-1})^\top \BS \Lambda^\top/N \BS f_\xi^{-1}(\omega)\|_l\\
&+\|\hat{\BS f}_\xi^{-1}(\omega) \boldsymbol{\hat \Lambda}-{\BS f}_\xi^{-1}(\omega) \BS \Lambda (\BS H_{NT}^{-1})\|_l \|G-\hat G\|_l \| \boldsymbol{\hat \Lambda}^\top/N \hat{\BS f}_\xi^{-1}(\omega)-(\BS H_{NT}^{-1})^\top \BS \Lambda^\top/N \BS f_\xi^{-1}(\omega)\|_l,
\end{align*}\normalsize
$G$ is of fixed dimension $r\times r$ and we first show that $\|G\|_l=O(1), l \in [1,\infty]$. For this, we have \small$$\|G\|_2\leq \Big(\sigma_{\min}((\BS H_{NT}^{-1})^\top \BS f_f^{-1}(\omega)\BS H_{NT}^{-1}/N)+\sigma_{\min}((\BS H_{NT}^{-1})^\top\BS \Lambda^\top \BS f_\xi^{-1}(\omega) \BS \Lambda \BS H_{NT}^{-1}/N)\Big)^{-1}.$$\normalsize Note that Lemma~\ref{lemma.1}, \ref{lem.1.part.0}) implies that $1/\sigma_{\min}(\BS H_{NT})=O(1)$ and $1/\sigma_{\min}(\BS H_{NT}^{-1})=O(1)$ and we have for symmetric matrices $A,B$, $1/\sigma_{\min}(AB)\leq 1/(\sigma_{\min} (A) \sigma_{\min}(B))$. Hence, $\sigma_{\min}((\BS H_{NT}^{-1})^\top f_f^{-1}(\omega)\BS H_{NT}^{-1}/N)=O(1/N)$. Furthermore, let $\tilde {\BS \Lambda}=(\BS \Lambda^\top \BS \Lambda/N)^{-1/2} \BS \Lambda$. Note that $\BS \Lambda^\top \BS \Lambda/N=\BS \Sigma_\Lambda+o(1)$ and $\BS \Sigma_\Lambda$ is positive definite by Assumption~\ref{ass.fac} and  $\sigma_{\min}(\BS \Lambda^\top \BS \Lambda/N)>1/M>0$. Then, $\tilde {\BS \Lambda}^\top \tilde {\BS \Lambda}/N=I_r$ and we have by Poincare's separation theorem $\sigma_{\min}((\BS H_{NT}^{-1})^\top\BS \Lambda^\top \BS f_\xi^{-1}(\omega) \BS \Lambda \BS H_{NT}^{-1}/N)\geq \sigma_{\min}(\BS H_{NT}^{-1})^2 \sigma_{\min}((\BS \Lambda^\top \BS \Lambda/N)^{-1}) \sigma_{\min}(\BS f_\xi^{-1}(\omega))$. Thus, $\|G\|_2=O(1)$ and since it is of fixed dimension, we also have $\|G\|_l=O(1), l \in [1,\infty]$.
Since $\BS f_X$ is hermitian, we can focus on $l=\infty$. We have by Assumption~\ref{ass.fac} and \ref{sparsity.b} $\| \BS f_\xi^{-1}(\omega) \BS \Lambda \BS H_{NT}^{-1}\|_\infty \leq \| \BS f_\xi^{-1}(\omega) \|_\infty \| \BS \Lambda \|_\infty \|\BS H_{NT}^{-1}\|_\infty \leq O(k_\xi)$. Note that $\BS \Lambda \in {N\times r}$ which means $\| \BS \Lambda \|_\infty\leq r \|\BS \Lambda\|_{\max}=O(1)$.\\

Similarly, since \small$\|\BS \Lambda^\top/N\|_\infty \leq N/N \|\BS \Lambda\|_{\max}$ \normalsize, we have \small$\|(\BS H_{NT}^{-1})^\top \BS \Lambda^\top/N \BS f_\xi^{-1}(\omega)\|_\infty=O(k_\xi)$.\normalsize By similar arguments, we have \small$\|\hat{\BS f}_\xi^{-1}(\omega) \boldsymbol{\hat\Lambda}-{\BS f}_\xi^{-1}(\omega) \BS \Lambda (\BS H_{NT}^{-1})\|_\infty=O_P(\|\hat{\BS f}_\xi^{-1}(\omega)-{\BS f}_\xi^{-1}(\omega)\|_\infty+k_\xi \|\boldsymbol{\hat\Lambda}-\BS \Lambda \BS H_{NT}^{-1}\|_{\max})=O_P(\|\hat{\BS f}_\xi^{-1}(\omega)-{\BS f}_\xi^{-1}(\omega)\|_\infty+k_\xi(\sqrt{\log(N)/T}+(NT)^{2/\zeta}/T+k_\xi/N))$\normalsize and \small$\| \boldsymbol{\hat\Lambda}^\top/N \hat{\BS f}_\xi^{-1}(\omega)-(\BS H_{NT}^{-1})^\top \BS \Lambda^\top/N \BS f_\xi^{-1}(\omega)\|_\infty\\=O_P(\|\hat{\BS f}_\xi^{-1}(\omega)-{\BS f}_\xi^{-1}(\omega)\|_\infty+k_\xi(\sqrt{\log(N)/T}+(NT)^{2/\zeta}/T+k_\xi/N))$.\normalsize 
We have further \small$\|G-\hat G\|_2 \leq \|G\|_2 \|\hat G\|_2 \|G^{-1}-\hat G^{-1}\|_2$\normalsize and  \small$\|G^{-1}-\hat G^{-1}\|_2\leq \|(\BS H_{NT}^{-1})^\top f_f^{-1}(\omega)\BS H_{NT}^{-1}/N-\hat f_f^{-1}(\omega)/N\|_2+\|(\BS H_{NT}^{-1}/\sqrt{N})^\top\BS \Lambda^\top \BS f_\xi^{-1}(\omega) \BS \Lambda \BS H_{NT}^{-1}/\sqrt{N}-\boldsymbol{\hat\Lambda}^\top/\sqrt{N} \hat{\BS f}_\xi^{-1}(\omega) \boldsymbol{\hat\Lambda})^{-1}/\sqrt{N}\|_2
$. \normalsize Note \small$\|(\BS H_{NT}^{-1}/\sqrt{N})^\top\BS \Lambda^\top-\boldsymbol{\hat\Lambda}^\top/\sqrt{N}\|_2\leq \|(\BS H_{NT}^{-1})^\top\BS \Lambda^\top-\boldsymbol{\hat\Lambda}^\top\|_{\max}=O_P(\sqrt{\log(N)/T}+(NT)^{2/\zeta}/T+k_\xi/N)$, $\|(\BS H_{NT}^{-1}/\sqrt{N})^\top\BS \Lambda^\top\|_2=O(1)$ and $\|\BS f_\xi^{-1} (\omega)\|_2=O(1)$. \normalsize\\ 

Hence, by these results and Lemma~\ref{lemma.spectral.density.factor} we have $\|G^{-1}-\hat G^{-1}\|_2=O_P(\sqrt{\log(N)/T}+(NT)^{2/\zeta}/T+k_\xi/N+ \|\BS f_\xi(\omega)^{-1}-\hat{\BS f}_\xi(\omega)^{-1}\|_2)=O_P(\sqrt{\log(N)/T}+(NT)^{2/\zeta}/T+k_\xi/N+\|{\BS \Sigma}_v^{-1}-\hat {\BS \Sigma}_v^{-1,CLIME}\|_2+k\max_s \|\hat \bbeta^{(s)}- \bbeta^{(s)}\|_2^{1-q})$ which is faster than $\|\hat{\BS f}_\xi^{-1}(\omega) \boldsymbol{\hat\Lambda}-{\BS f}_\xi^{-1}(\omega) \BS \Lambda (\BS H_{NT}^{-1})\|_\infty$. That means 
$\|\BS f_X(\omega)^{-1}-\boldsymbol{\hat f}_X(\omega)^{-1}\|_\infty=O_P(k_\xi\|\hat{\BS f}_\xi^{-1}(\omega)-{\BS f}_\xi^{-1}(\omega)\|_\infty+k_\xi^2 \|\boldsymbol{\hat\Lambda}-\BS \Lambda \BS H_{NT}^{-1}\|_{\max})$.
Since $\| \BS \Lambda \BS H_{NT}^{-1}/\sqrt{N}\|_2=O(1)$ and $\|\BS f_\xi(\omega)\|_2=O(1)$, we have further $\|\BS f_X(\omega)^{-1}-\boldsymbol{\hat f}_X(\omega)^{-1}\|_2=O_P(\|\hat{\BS f}_\xi^{-1}(\omega)-{\BS f}_\xi^{-1}(\omega)\|_2+\|\boldsymbol{\hat\Lambda}-\BS \Lambda \BS H_{NT}^{-1}\|_{\max})$. The assertions follows after inserting the rates of Lemma~\ref{lem.rate.spectral.var}. 
\end{proof}
\end{appendices}
\newpage
\section*{Supplementary Material}
In this supplementary material we collect the proofs of the lemmas.
\begin{proof}[Proof of Lemma~\ref{lemma.1}]
First note that under Assumption~\ref{ass.moments} and Remark~\ref{remark.functional.dependence} we have
for $\alpha>0.5, q\geq4$ $\max_{j} \|\{\xi_{j,t}\}\|_{q,\alpha}<\infty,\max_{j} \|\{f_{i,t}\}\|_{q,\alpha}<\infty$ and $\max_{j,i} \|\{\xi_{j,t} f_{i,t}\}\}\|_{q/2,\alpha}<\infty$. Furthermore, since $\{\BS f_t\}$ and $\{\BS \xi_t\}$ are linear processes and $\|\BS B^{(j)}\|_2\leq M \rho^j$, we have $\max_{\BS w \in \R^n, \|\BS w\|_2\leq 1} \|\{ \BS w^\top \BS \xi_{t}\}\|_{q,\alpha}<\infty=
\max_{\BS w \in \R^n, \|\BS w\|_2\leq 1} \| \{\BS w ^\top \BS B^{(j)} \BS \eps_t\} \|_{q,\alpha}=
\max_{\BS w \in \R^n, \|\BS w\|_2\leq 1} \| \sup_{m\geq 0} (m+1)^\alpha \sum_{k=m}^\infty \|\BS w^\top \BS B^{(k)} (\eps_0-\eps_0^\prime)\|_{E,q} \leq \sup_{m\geq 0} (m+1)^\alpha C \rho^m<\infty
$ and similarly $\max_{j} \|\{\xi_{j,t}^2\}\|_{q/2,\alpha}<\infty$ and $\max_{j} \|\{f_{i,t}^2\}\|_{q/2,\alpha}<\infty$.\\ 

For part \ref{lem.1.part.0}, we have $\BS X^\top \BS X/NT=\BS \Lambda \BS F^\top \BS F \BS \Lambda^\top/NT+\BS \Xi^\top \BS \Xi/NT+\BS \Xi^\top \BS F \BS \Lambda^\top/NT+\BS \Lambda \BS F^\top \BS \Xi/NT$. Assumption~\ref{ass.fac}, i.e, $\BS \Sigma_\Lambda>0, \BS \Sigma_F>0$, implies for $N,T$ large enough that $\BS F/\sqrt{T}$ and $\BS \Lambda/\sqrt{N}$ have rank $r$ and  that all $r$ eigenvalues are strictly positive. Furthermore, note that we have by part~\ref{lem.1.part.C} and Assumption~\ref{sparsity} $\|\BS \Xi^\top \BS \Xi/NT\|_F^2=(1/N^2 \sum_{i_1,i_2}^N (1/T \sum_{t=1}^T \xi_{i_1,t} \xi_{i_2,t})^2)=(1/N^2 \sum_{i_1,i_2}^N (\BS e_{i_1}^\top \BS \Gamma_\xi(0) \BS e_{i_2}+1/T \sum_{t=1}^T \xi_{i_1,t} \xi_{i_2,t}-\BS e_{i_1}^\top \BS \Gamma_\xi(0) \BS e_{i_2})^2)=O_P(k_\xi/N+\log(N)/T+(NT)^{1/\zeta}/T^2)$. Additionally, by part~\ref{lem.1.part.A} we have $\|\BS \Lambda \BS F^\top \BS \Xi/NT\|_F^2=\|\BS \Xi^\top \BS F \BS \Lambda^\top/NT\|_F^2=1/N^2 \sum_{i_1,i_2}^N \sum_{l=1}^r (1/T \sum_{t=1}^T \xi_{i_1,t}f_{r,t})^2 \ell_{i,l}=O_P(\log(N)/T+(NT)^{1/\zeta}/T^2)$. That means for $N,T$ large the eigenvalues of $\BS X^\top \BS X/NT$ are approximately those of $\BS \Lambda \BS F^\top \BS F \BS \Lambda^\top/NT$. Hence, for $N,T$ large, $\BS X^\top \BS X/NT$ possesses $r$ positive eigenvalues which implies that $\BS D_{NT,r}^2$ is invertible and consequently, $\BS D_{NT,r}^{-2}=O_P(1)$. Since by Assumption~\ref{ass.fac} $\lim_T\|\BS F/\sqrt{T}\|_2\leq M$ and $\lim_T\|\BS \Lambda/\sqrt{N}\|_2\leq M$, we also have $\BS D_{NT,r}^{2}=O_P(1)$.\\ 

For the part \ref{lem.1.part.A}, note first that $\mathbb{E}\frac{1}{T}\sum_{s=1}^T \xi_{i,s} f_{j,s}=0$ due to Assumption \ref{ass.moments}. Furthermore, since $P(\max_{i,j}|\sum_{s=1}^T \xi_{i,s} f_{j,s}| \geq x) \leq \sum_{i,j} P(|\sum_{s=1}^T \xi_{i,s} f_{j,s}| \geq x)$, we have Assumption \ref{ass.moments} and Theorem~2 in \cite{wu2016performance} for $q >2$ and $C_1,C_2,C_3$ some constants depending only on $q$ and $\alpha$
\small
$$
P\left(\left|\sum_{s=1}^T \xi_{i,s} f_{j,s}\right| \geq x\right)
\leq C_1 \frac{T \max_{j,i} \|\{\xi_{j,t} f_{i,t}\}\|_{q,\alpha}}{x^q}+C_2 \exp\left(-\frac{C_3 x^2}{T \max_{j,i} \|\{\xi_{j,t} f_{i,t}\}\|_{2,\alpha}^2}\right).
$$\normalsize
Furthermore, 
\small
$$
P\left(\max_{i,j}\left|\sum_{s=1}^T \xi_{i,s} f_{j,s}\right| \geq x\right)
\leq C_1 \frac{NT \max_{j,i} \|\{\xi_{j,t} f_{i,t}\}\|_{q,\alpha}}{x^q}+C_2 \exp\left(-\frac{C_3 x^2}{T \max_{j,i} \|\{\xi_{j,t} f_{i,t}\}\|_{2,\alpha}^2}+\log(N)\right).
$$\normalsize
This implies $\max_{i,j} |1/T\sum_{s=1}^T \xi_{i,s} f_{j,s}|=O_P(\sqrt{(\log(N)/T})+N^{2/\zeta}T^{2/\zeta-1})$. Since $\mathbb{E} f_{j,t}^2=\Sigma_{f,j,j}$, $\mathbb{E} \xi_{j_1,t} \xi_{j_2t}=\BS e_{j_1}^\top \BS \Gamma_{\xi} \BS e_{j_2}$, and $r$ is fixed, Part \ref{lem.1.part.B} and \ref{lem.1.part.C} follow by the same arguments. Note also that for some vectors $\BS u,\BS  v$ and some symmetric matrix $\BS \Gamma$, we have $\BS u^\top \BS \Gamma \BS u \leq \BS v^\top \BS \Gamma \BS v + \BS u^\top \BS \Gamma \BS u$. That is why for  $\max_{j_1,j_2}|\sum_{s=1}^T (\xi_{j_1,s} \xi_{j_2,s}-\BS e_{j_1}^\top \BS \Gamma_{\xi} \BS e_{j_2})$ it is sufficient to look at $\max_{j}|\sum_{s=1}^T (\xi_{j,s} \xi_{j,s}-\BS e_{j}^\top \BS \Gamma_{\xi} \BS e_{j})$.\\  

For the part \ref{lem.1.part.D}, note that $\BS \Gamma_{\xi}(0)=\sum_{j=0}^\infty \BS B^{(j)} \BS \Sigma_v (\BS B^{(j)})^\top$ and $1/\sqrt{N} \sum_{i=1}^N \ell_{i,k} \xi_{i,s}=1/\sqrt{N} \BS e_k^\top \BS \Lambda \BS \xi_s$, where $1/\sqrt{N} \BS e_k^\top \BS \Lambda\in \R^N$ and $\|1/\sqrt{N} \BS e_k^\top \BS \Lambda\|_2\leq M$. Since $\|\BS B^{(j)}\|_2\leq M \rho^j$ and $\|\BS \Sigma_v\|_2\leq M$ by Assumption~\ref{sparsity},\ref{ass.moments}, we have $\|\BS \Gamma_{\xi}(0)\|_2 \leq M^3/(1-\rho^2)$ and the assertions follows then by Assumption~\ref{ass.fac} and part \ref{lem.1.part.C}. For Part \ref{lem.1.part.E} we obtain by similar arguments for $q>2$,
\small
\begin{equation*}
\begin{aligned}
&P\left(\max_{j,k}\left|\frac{1}{T}\sum_{t=1}^T \BS e_j^\top \BS \xi_t \BS \xi_t^\top \BS \Lambda^\top \BS e_k- \BS e_j^\top \BS \Gamma_\xi(0) \BS \Lambda^\top \BS e_k\right| \geq x\right)\leq C_1 \frac{NT \max_{i} \|\{\xi_{i,t}^2\}\|_{q,\alpha}}{x^q}+\\
&+C_2 \exp\left(-\frac{C_3 x^2}{T \max_{i} \|\{\xi_{i,t}^2\}\|_{2,\alpha}^2}+\log(N)\right).
\end{aligned}
\end{equation*}\normalsize
Since $\|\BS \Gamma_\xi(0) \BS \Lambda^\top\|_{\max}\leq\|\BS \Gamma_\xi(0)\|_\infty \|\BS \Lambda\|_{\max}\leq M^2 k_\xi$, we have 
$\max_{j,k} |1/T\sum_{s=1}^T \BS e_j^\top \BS \xi_t \BS \xi_t^\top \BS \Lambda^\top \BS e_k|=O_P(k_\xi+\sqrt{\log(N)/T}+N^{2/\zeta}T^{2/\zeta-1})$.
.\\

For Part~\ref{lem.1.part.F}, first note that we have by Cauchy-Schwarz and the previous parts of this lemma
\small
\begin{align*}
    \left|\frac{1}{NT^2}\sum_{i=1}^N \sum_{s,t=1}^T \xi_{j,t} \BS f_t^\top \BS \Lambda_i \xi_{i,s} \hat{f}_{s,l}\right|&\leq 
    \left(\sum_{k=1}^r \left(\frac{1}{T} \sum_{t=1}^T \xi_{j,t} f_{k,t}\right)^2\right)^{1/2} \left(\sum_{k=1}^r \left(\frac{1}{N^2T} \sum_{s=1}^T (\BS e_k^\top \BS \Lambda^\top \BS \xi_s)^2\right)\left(\frac{1}{T} \sum_{s=1}^T \hat f_{l,s}^2\right)\right)^{1/2} \\
    &=\frac{1}{\sqrt{NT}}
    \left(\sum_{k=1}^r \left(\frac{1}{\sqrt{T}} \sum_{t=1}^T \xi_{j,t} f_{k,t}\right)^2\right)^{1/2} \left(\sum_{k=1}^r \left(\frac{1}{NT} \sum_{s=1}^T (\BS e_k^\top \BS \Lambda^\top \BS \xi_s)^2\right)\right)^{1/2}\\
    &=O_P\left(\frac{1}{\sqrt{NT}}\right),
\end{align*}\normalsize
$\max_{j}|\frac{1}{NT^2}\sum_{i=1}^N \sum_{s,t=1}^T \xi_{j,t} \BS f_t^\top \BS \Lambda_i \xi_{i,s} \hat{f}_{s,l}|=O_P(\sqrt{\log(N)}/\sqrt{NT}+N^{2/\zeta-1/2}T^{2/\zeta-3/2})$,
\small
\begin{align*}
   &\max_{j}\left|\frac{1}{NT^2}\sum_{i=1}^N \sum_{s,t=1}^T \xi_{j,t} \xi_{i,t} \BS \Lambda_i \BS f_s \hat f_{l,s}\right|\leq\\ 
    &\leq\left(\sum_{k=1}^r \max_{j} \left(\frac{1}{NT} \sum_{t=1}^T \BS e_j^\top \BS \xi_t \BS \xi_t^\top \BS \Lambda^\top e_k\right)^2\right)^{1/2}\left(\sum_{k=1}^r \left(\frac{1}{T} \sum_{s=1}^T f_{k,s}^2\right)\left(\frac{1}{T} \sum_{s=1}^T \hat f_{l,s}^2\right)\right)^{1/2} \\
    &\leq \frac{k_\xi}{N} \left(\sum_{k=1}^r \max_{j} \sum_{t=1}^T \left(\frac{1}{k_\xi T} \BS e_j^\top \BS \xi_t \BS \xi_t^\top \BS \Lambda^\top e_k\right)^2\right)^{1/2}\left(\sum_{k=1}^r \left(\frac{1}{T} \sum_{s=1}^T f_{k,s}^2\right)\right)^{1/2}\\
    &=O_P\left(\frac{k_\xi}{N}+\frac{\sqrt{\log(N)}}{\sqrt{T}N}+N^{2/\zeta-1}T^{2/\zeta-1}\right),
\end{align*}
\begin{align*}
    \max_{j}&\left|\frac{1}{NT^2} \sum_{i=1}^N \sum_{s,t=1}^T \xi_{j,t} \xi_{i,t} \xi_{i,s} f_{l,s}\right|=\max_j\left|
    \frac{1}{N}\sum_{i=1}^N \left(\frac{1}{T} \xi_{j,t} \xi_{i,t} - \BS e_j \BS\Gamma_\xi(0) \BS e_i+\BS e_j \BS \Gamma_\xi(0) \BS e_i\right) \left(\frac{1}{T} \sum_{s=1}^T \xi_{i,s} f_{l,s}\right)\right| \\
    =&\left\|\frac{1}{T} \sum_{t=1}^T \BS \xi_t \BS \xi_t^\top -\BS \Gamma_\xi(0)\right\|_{\max} \left\|\frac{1}{T} \sum_{s=1}^T \BS \xi_s f_{l,s}\right\|_{\max}+\left\|\frac{1}{T} \sum_{s=1}^T \BS \xi_s f_{l,s}\right\|_{\max}\left\|\BS \Gamma_\xi(0)\right\|_\infty/N\\
    =&O_P\left(\frac{\log(N)}{T}+(NT)^{4/\zeta}/T^2+\frac{k_\xi}{N}\left(\sqrt{(\log(N)/T})+N^{2/\zeta}T^{2/\zeta-1}\right)\right),
\end{align*}
\normalsize
and \small
\begin{align*}
    \max_j\left|\frac{1}{NT^2} \sum_{i=1}^N \sum_{s,t=1}^T \xi_{j,t} \xi_{i,t} \xi_{i,s} \hat f_{l,s}\right|&=O_P\left(\frac{{\log(N)}}{{T}}+N^{4/\zeta}T^{4/\zeta-2}+\frac{k_\xi}{N}+\frac{\sqrt{\log(N)}}{\sqrt{NT}}+N^{2/\zeta-1/2}T^{2/\zeta-1}\right).
\end{align*}\normalsize
Then, we have further
\small
\begin{align*}
    &\frac{1}{NT^2} \sum_{i=1}^N \sum_{s,t=1}^T \xi_{j,t} \xi_{i,t} \xi_{i,s} \hat f_{l,s}=\\
    &=\sum_{k=1}^r \left(\frac{1}{NT^2} \sum_{i=1}^N \sum_{s,t=1}^T \xi_{j,t} \xi_{i,t} \xi_{i,s} f_{k,s}\right)\left(\BS e_l^\top \BS H_{NT} \BS e_k\right)+
    \frac{1}{NT^2} \sum_{i=1}^N \sum_{s,t=1}^T \xi_{j,t} \xi_{i,t} \xi_{i,s} (\hat f_{l,s}- \BS e_l \BS H_{NT} \BS f_s)\\
    &=I_{j}+II_{j},
\end{align*}
$\max_j|I_{j}|=O_P\left(\frac{\sqrt{\log(N)}}{T}+N^{2/\zeta}T^{2/\zeta-3/2}+\frac{k_\xi}{N}\left(\sqrt{(\log(N)/T})+N^{2/\zeta}T^{2/\zeta-1}\right)\right)$. Furthermore, we have
\begin{align*}
    &\max_j |II_j|=\\
    &=\max_j\left|\frac{1}{N^2T^3} \sum_{k=1}^N \sum_{s,t=1}^T \xi_{j,s} \xi_{k,s} \xi_{k,t} \left[ \sum_{i=1}^N \sum_{s=1}^T \BS f_t^\top \BS\Lambda_i \xi_{i,s} \hat{\BS f_s}+ \sum_{i=1}^N \sum_{s=1}^T \xi_{i,t} \BS\Lambda_i \BS f_{s}^\top \hat{\BS f_s}+ \sum_{i=1}^N \sum_{s=1}^T \xi_{i,t}\xi_{i,s}  \hat{\BS f_s}\right] \BS D_{NT,r}^{-2} \BS e_l\right| \\
    \leq&
    \frac{1}{\sqrt{N}}\max_{j,i} \left|\frac{1}{T} \sum_{s=1}^T \xi_{j,s} \xi_{i,s}\right| \max_{j,l} \left|\frac{1}{T} \sum_{t=1}^T \xi_{j,t} f_{l,t}\right|  \left( \max_l\frac{1}{NT}  \sum_{s=1}^T ( \BS e_l^\top \BS \Lambda \BS \xi_s)^2\right)^{1/2} \|\BS D_{NT}^{-2}\|_{\max}r^2 \\
    &+\max_{j,i} \left|\frac{1}{T} \sum_{s=1}^T \xi_{j,s} \xi_{i,s}\right| \max_{j,l}\left|\frac{1}{NT} \sum_{t=1}^T \xi_{j,t} \BS \xi_t^\top \BS \Lambda \BS e_l\right| \max_{k,l} \left|\frac{1}{T} \sum_{s=1}^T f_{k,s} \hat f_{l,s}\right| \|\BS D_{NT}^{-2}\|_{\max} r^2 \\
    &+\max_j \frac{1}{N^2}\left|\BS e_j^\top \left(\frac{1}{T} \sum_{t=1}^T \BS \xi_t \BS \xi_t^\top \right)^5  \BS e_j\right|^{1/2} \\
    =&O_P\left(1+\left(\sqrt{\log(N)/T}+N^{2/\zeta}T^{2/\zeta-1}\right)/\sqrt{N}\left(\sqrt{\log(N)/T}+N^{2/\zeta}T^{2/\zeta-1}\right)\right)\\
    &+O_P\left(1+\left(\sqrt{\log(N)/T}+N^{2/\zeta}T^{2/\zeta-1}\right)\left(\frac{k_\xi}{N}+\frac{\sqrt{\log(N)}}{\sqrt{T}N}+N^{2/\zeta-1}T^{2/\zeta-1}\right)\right)
    +III\\
    =&O_P\left(\frac{{\log(N)}}{{T}}+N^{4/\zeta}T^{4/\zeta-2}+\frac{k_\xi}{N}+\frac{\sqrt{\log(N)}}{\sqrt{NT}}+N^{2/\zeta-1/2}T^{2/\zeta-1}\right),
\end{align*}\normalsize
where $III\leq 1/N^2 (\|(\frac{1}{T} \sum_{t=1}^T \BS \xi_t \BS \xi_t^\top)-\BS \Gamma_\xi(0)\|_\infty +\|\BS \Gamma_\xi(0)\|_\infty)^2 \|\frac{1}{T} \sum_{t=1}^T \BS \xi_t \BS \xi_t^\top\|_{\max}^{1/2}\leq (\|(\frac{1}{T} \sum_{t=1}^T \BS \xi_t \BS \xi_t^\top)-\BS \Gamma_\xi(0)\|_{\max}^2+k_\xi/N)^2 \|\frac{1}{T} \sum_{t=1}^T \BS \xi_t \BS \xi_t^\top\|_{\max}^{1/2} =\\O_P((\frac{\sqrt{\log(N)}}{\sqrt{T}}+N^{2/\zeta}T^{2/\zeta-1}+k_\xi/N)^2(1+(\sqrt{\log(N)/T}+N^{2/\zeta}T^{2/\zeta-1})))$. Then, the assertion is the combination of the previous results.\\

For part~\ref{lem.1.part.G}, first note that we have from part~\ref{lem.1.part.F}. 
\small
\begin{align*}
    \Big|\frac{1}{NT}\sum_{i=1}^N& \sum_{s=1}^T \xi_{i,t} \xi_{i,s}[\hat{\BS f_s} - \BS  H_{NT} {\BS f_s}]\Big|\leq \frac{1}{N} \sum_{s=1}^N |\xi_{i,k}| \max_{j,l} \left|1/T\sum_{s=1}^T \BS  e_l^\top(\hat{\BS f}_s - \BS  H_{NT} {\BS f}_s) \xi_{i,s}\right| \\
    &=O_P\left(\frac{{\log(N)}}{{T}}+\frac{k_\xi}{N}+\frac{\sqrt{\log(N)}}{\sqrt{NT}}+\gt\right).
\end{align*}
\normalsize
Furthermore, from part~\ref{lem.1.part.A} and Assumption~\ref{ass.fac}. we have
\begin{align*}
    &\left|\frac{1}{NT}\sum_{i=1}^N \sum_{s=1}^T \BS  e_l^\top \BS f_t^\top \BS\Lambda_i \xi_{i,s}[\hat{\BS f}_s - \BS  H_{NT} {\BS f_s}]\right|=\\
    &=O_P\left(\frac{{\log(N)}}{{T}}+\frac{k_\xi}{N}+\frac{\sqrt{\log(N)}}{\sqrt{NT}}+\gt\right).
\end{align*}
\normalsize
Then, the assertion follows by the same arguments as in part~\ref{lem.1.part.F}. Note that in each assertion in part~\ref{lem.1.part.F} $\xi_{j,t}$ is replaced with $\|\BS  e_l^\top \BS \Lambda^\top \BS\|_2/N(\BS  e_l^\top \BS \Lambda^\top \BS/\|\BS  e_l^\top \BS \Lambda^\top \BS\|_2 \xi_t)$. Since $\|\BS  e_l^\top \BS \Lambda^\top \BS\|_2/N=O(1/\sqrt{N)}$ the assertion follows by part~\ref{lem.1.part.A},\ref{lem.1.part.D}.\\

For part~\ref{lem.1.part.H}, we have for $j,l=1,\dots,r$\small
\begin{align*}
     \Bigg|\frac{1}{T}\sum_{s=1}^T& f_{j,s} [\hat{\BS f}_s - \BS  H_{NT} \BS f_s]^\top \BS D_{NT,r}^2 \BS  e_l \Bigg|=\frac{1}{NT^2} \left|\sum_{t,s=1}^T \sum_{i=1}^N f_{j,t} \xi_{i,t} \BS \Lambda_i \BS f_s^\top \hat f_{s,l} +\sum_{t,s=1}^T f_{j,t} \BS \xi_{t}^\top \BS\xi_{s} f_{s,l} \right|\\
    &+O_P \left(\frac{\log(N)}{T}+\frac{k_{\xi}}{N}+\frac{\sqrt{\log(N)}}{\sqrt{NT}}+(NT)^{2/\zeta} \left(\frac{1}{\sqrt{N}T}+\frac{1}{T^{3/2}}+(NT)^{2/\zeta}\frac{1}{T^2} \right)
    \right) \\
    \leq& \left(\sum_{k=1}^r \left(\frac{1}{TN} \sum_{t=1}^T (\BS  e_k^\top \BS \Lambda^\top \BS \xi_t f_{j,t})^2\right)^{1/2} \left(\sum_{k=1}^r \left(\frac{1}{T}\sum_{s=1}^T f_{k,r}^2 \right)^{1/2} \right) \right) \\
    &+ \left[ \left(\frac{1}{T} \sum_{t=1}^T f_{j,t} \BS \xi_t^\top \right) \left(\frac{1}{T} \sum_{s=1}^T \BS \xi_s \xi_s^\top \right) \left(\frac{1}{T} \sum_{z=1}^T \BS \xi_z  f_{j,z} \right) \right]^{1/2}\frac{1}{N}\\
    &+O_P \left(\frac{{\log(N)}}{{T}}+\frac{k_{\xi}}{N}+\frac{\sqrt{\log(N)}}{\sqrt{NT}}+(NT)^{2/\zeta} \left(\frac{1}{\sqrt{N}T}+\frac{1}{T^{3/2}}+(NT)^{2/\zeta}\frac{1}{T^2} \right)
    \right)\\
    &=O_P \left(\frac{\log(N)}{T}+\frac{k_{\xi}}{N}+\frac{\sqrt{\log(N)}}{\sqrt{NT}}+(NT)^{2/\zeta} \left(\frac{1}{\sqrt{N}T}+\frac{1}{T^{3/2}}+(NT)^{2/\zeta}\frac{1}{T^2} \right) \right).
\end{align*}\normalsize
Part~\ref{lem.1.part.I} follows by part~\ref{lem.1.part.A},\ref{lem.1.part.D}, and \ref{lem.1.part.G}. For Part~\ref{lem.1.part.J}, note that $\|\BS  \Lambda\|_{\max}\leq M$. Then, this part follows by \ref{lem.1.part.F},\ref{lem.1.part.H}, and \ref{lem.1.part.I}. 
\end{proof}

\begin{proof}[Proof of Lemma~\ref{lem_spars}]
This proof follows ideas of the Proof of Proposition 4.1 in \cite{basu2015} as well as the Proof of Corollary 3 in \cite{negahban2012unified}. Let $\hat \gamma=1/(T-p) \sum_{t=p+1}^T \hat \xi_{j,t} \boldsymbol{\hat \xi}_{t-1}^v$ and $\boldsymbol{\hat \Gamma}=1/(T-p) \sum_{t=p+1}^T \boldsymbol{\hat \xi}_{t-1}^v (\boldsymbol{\hat \xi}_{t-1}^v)^\top$. Let $\BS \beta^*:=\BS \beta_j, \boldsymbol{\hat \beta}:=\boldsymbol{\hat \beta}_j$ and $\BS v=\boldsymbol{\hat \beta}-\BS \beta^*$. Furthermore, let for some threshold $\eta>0$ $J=J_\eta=\{j\in\{1,\dots,np\} | \BS e_j^\top \BS \beta^*\|>\eta\}$ denote the set of indices for which $\bbeta^*$ is absolutely greater than the threshold $\eta$, $\bbeta_nu$ refers to the hard thresholded vector with threshold $\eta$ and for some vector $\bu$, $\bu_J, \bu_{J^C}$ denotes the vector obtained by the indices in $J$, $J^C$, respectively.
We have by Assumption~\ref{sparsity} $\|\bbeta^*_\eta-\bbeta^*\|_1\leq \eta^{1-q} k$. Furthermore, $|J|\leq \eta^{-q} k$.\newline Since $\hat \bbeta_j$ is the minimum given in \eqref{eq.lasso.beta.j}, we have $-\hat\bbeta^\top \hat \bgamma+\hat \bbeta^\top \hat \bGamma \hat \bbeta+\hat \blambda \|\hat \bbeta\|_1 \leq 2 \bbeta^* \hat \bgamma+(\bbeta^*)^\top \hat \bGamma \bbeta^*+\hat \blambda\|\bbeta^*\|_1$. This gives further $\bv^\top \hat \bGamma \bv\leq 2 v^\top (\hat \bgamma-\hat\bGamma \bbeta^*)+\hat \blambda(\|\bbeta*\|_1-\|\bbeta^*+\bv\|_1)\leq 2 \bv^\top (\hat \bgamma-\hat\bGamma \bbeta*)+\hat \blambda(\|\bbeta_\eta^*\|_1+2\|\bbeta^*-\bbeta_\eta^*\|_1-\|\bbeta_\eta^*+\bv\|_1)\leq 2 \bv^\top (\hat \bgamma-\hat\bGamma \bbeta^*)+\hat \blambda(\|\bv_J\|_1-\|\bv_{J^C}\|_1+2 \eta^{1-q}k)$.
This implies with the condition $\|\frac{1}{T-p} \sum_{t=p+1}^T  (\hat \xi_{j,t}- \bbeta_j^\top \hat \bxi_{t-1}^v) \hat \bxi_{t-1}^v\|_{\max} \leq \hat \blambda/4 $ that $0\leq \bv^\top \hat \bGamma v\leq 3/2 \hat \blambda \|\bv_J\|_1-1/2 \|\bv_{J^C}\|_1\leq 2 \hat \blambda \|\bv\|_1+2\hat \blambda \eta^{1-q}k$. Hence, $\|\bv_{J^C}\|_1\leq 3 \|\bv_J\|_1+4 \eta^{1-q} k$ and since $|J|\leq \eta^{-q} k$, $\|\bv\|_1\leq 4 \sqrt{k} \eta^{-q/2} \|\bv\|_2+4s \eta^{1-q}$. 

Then, with the condition $\bTheta^\top \frac{1}{T-p} \sum_{t=p+1}^T \hat \bxi_{t-1}^v (\hat \bxi_{t-1,j})^\top \bTheta\geq \alpha \|\bTheta\|_2^2-\hat\tau \|\bTheta\|_1^2 \, \forall\, \bTheta \in \R^{np}$  we obtain that $\alpha\|\bv\|_2^2-\hat \tau \|\bv\|_1 \leq 8 \hat \blambda \|\bv\|_2 \sqrt{k} \eta^{-q/2}+10 \hat \blambda k \eta^{1-q}$. Set $\eta=\hat \blambda/\alpha$. Then, with the bound for $\|\bv\|_1$ and dropping minor terms in the maximum we obtain 
$\|\bv\|_2\leq 16 \max(\sqrt{k} (\hat \blambda/\alpha)^{1-q/2},\sqrt{\hat \tau} s (\hat \blambda/\alpha)^{1-q})$. Furthermore, $\|v\|_1\leq \max(68 k (\hat \blambda/\alpha)^{1-q},64 \sqrt{\tau}k^{3/2} (\hat \blambda/\alpha)^{1-3/2q}+4k(\hat \blambda/\alpha)^{1-q}).$
\end{proof}

\begin{proof}[Proof of Lemma~\ref{lemma.spectral.density.factor}]
We have $\|\boldsymbol{\hat f}_f(\omega)-\BS H_{NT}\BS f_f(\omega)\BS H_{NT}^\top\|_{\max}\leq
\|\boldsymbol{\hat f}_f(\omega)-\BS H_{NT}\boldsymbol{\tilde f}_f(\omega)\BS H_{NT}^\top\|_{\max}+\|\boldsymbol{\tilde f}_f(\omega)-\BS f_f(\omega)\|_{\max}\|\BS H_{NT}\|_1 \|\BS H_{NT}\|_\infty$. We have by Lemma~\ref{lemma.1} \ref{lem.1.part.0}) $\|\BS H_{NT}\|_1=O(1)=\|\BS H_{NT}\|_\infty$ and by Theorem 3 in \cite{wu2018asymptotic} $\|\boldsymbol{\tilde f}_f(\omega)-\BS f_f(\omega)\|_{\max}=O_P(\sqrt{B_T/T})$. Note that the dimension $r$ of the process $\{\BS f_t\}$ is fixed. Furthermore, we have 
$\|\boldsymbol{\hat f}_f(\omega)-\BS H_{NT}\BS f_f(\omega)\BS H_{NT}^\top\|_{\max}=\\
\|1/(2\pi) \sum_{h=-T+1}^{T-1} K\left(\frac{u}{B_T}\right) \exp(-ih \omega) (\hat \Gamma_f(h)-\BS H_{NT} \tilde \Gamma_f(h) \BS H_{NT}^\top)\|_{\max}$. Additionally, we have by Lemma~\ref{lemma.1} \ref{lem.1.part.H}),\ref{lem.1.part.I}) $\|\hat \Gamma_f(h)-\BS H_{NT} \tilde \Gamma_f(h) \BS H_{NT}\|_{\max}=\|\BS H_{NT}\|_\infty \|1/T \sum_t f_{t+h} [\boldsymbol{\hat f}_t-\BS H_{NT} \BS f_t]^\top\\+\|1/T\sum_t [\boldsymbol{\hat f}_{t+h}-\BS H_{NT} f_{t+h}] \BS f_t^\top\|_{\max}\|\BS H_{NT}\|_1+
\|1/T \sum_{t} [\boldsymbol{\hat f}_{t+h}-\BS H_{NT} f_{t+h}][\boldsymbol{\hat f}_t-\BS H_{NT} \BS f_t]^\top\\=O_P\left(\frac{{\log(N)}}{{T}}+\frac{k_\xi}{N}+\frac{\sqrt{\log(N)}}{\sqrt{NT}}+\gt\right).$ For the second assertion, note that $A^{-1}-B^{-1}=B^{-1}(B-A)A^{-1}$ and since the dimension of $\BS f_f$ is fixed, the second assertion follows immediately. 
\end{proof}

\begin{proof}[Proof of Lemma~\ref{lem.rate.spectral.var}]
First consider the estimation error in the residuals. For this, we consider the (unfeasible) sample covariance $\tilde {\BS \Sigma}_v=1/T \sum_t \BS v_t \BS v_t^\top$. For this, we have based on $\zeta$ moments and the Fuk-Nagaev inequality
$\|\tilde {\BS \Sigma}_v-\BS \Sigma_v\|_{\max}=O_P(\sqrt{(\log(N)/T})+N^{2/\zeta}T^{2/\zeta-1})$. Note that we have only the estimated residuals given by $\hat{\BS v}_t=\hat{\BS \xi}_t-\sum_{j=1}^p \hat{\BS A}^{(j)} \hat{\BS \xi}_{t-j}$. This gives the sample covariance 
$\hat {\BS \Sigma}_v=1/T \sum_t \hat{\BS v}_t \hat {\BS v}_t^\top$. We have $\tilde {\BS \Sigma}_v- \hat {\BS \Sigma}_v=1/T \sum_t ({\BS v}_t-\hat{\BS v}_t) {\BS v}_t^\top+ \BS v_t ({\BS v}_t-\hat{\BS v}_t)^\top +({\BS v}_t-\hat{\BS v}_t)({\BS v}_t-\hat{\BS v}_t)^\top$. Furthermore, ${\BS v}_t-\hat{\BS v}_t=\BS w_t+\sum_{j=1}^p \BS A^{(j)} \BS w_{t-j}+\sum_{j=1}^p(\hat{\BS A}^{(j)}-\BS^{(j)}) \BS \xi_{t-j}+\sum_{j=1}^p (\hat{\BS A}^{(j)}-\BS^{(j)}) \BS w_{t-j}$. Hence, following the arguments of Theorem~\ref{thm.est.error.step.one} we have $\|\tilde {\BS \Sigma}_v- \hat {\BS \Sigma}_v\|_{\max}=O_P(\|\A\|_\infty \|1/T \sum_{t=1} \BS w_t \BS \xi_t\|_{\max} +\max_j \|\hat \bbeta^{(j)}- \bbeta^{(j)}\|_1 ( \|1/T\sum_t \BS \xi_{t-1}^v \BS v_t\|_{\max}+\|1/T \sum_{t=1} \BS w_t \BS \xi_t\|_{\max})).$ Since $\BS \xi_{t-1}^v \BS v_t=0$, we have by the arguments of Lemma~\ref{lemma.1} $\|1/T\sum_t \BS \xi_{t-1}^v \BS v_t\|_{\max}=O_P(\sqrt{(\log(Np)/T})+(Np)^{2/\zeta}T^{2/\zeta-1})$. Together with Theorem~\ref{thm.est.error.step.one} and $\|\tilde {\BS \Sigma}_v-\BS \Sigma_v\|_{\max}$  this lead to the following\small
\begin{align*}
\|{\BS \Sigma}_v- \hat {\BS \Sigma}_v\|_{\max}=&O_P\Bigg(\sqrt{(\log(N)/T})+N^{2/\zeta}T^{2/\zeta-1}+k\Big[\frac{k_\xi}{N}+\frac{\log(N)}{T}+\frac{\sqrt{\log(N)}}{\sqrt{NT}}+(NT)^{2/\zeta-1}k_\xi+\frac{(NT)^{4/\zeta}}{T^2}\Big]+ \\
&(\sqrt{(\log(Np)/T})+(Np)^{2/\zeta}T^{2/\zeta-1})\Bigg(k\Bigg[\sqrt{\log(Np)/T}+(NpT)^{2/\zeta}/T+k\Bigg(\frac{k_\xi}{N}+\frac{\sqrt{\log(Np)}}{\sqrt{NT}}\\
&+(NpT)^{2/\zeta}\Bigg(\frac{k_\xi}{NT}+\frac{1}{\sqrt{N}T}+\frac{1}{T^{3/2}}+(NpT)^{2/\zeta}\frac{1}{T^2}\Bigg)\Bigg)\Bigg]^{1-q}\Bigg) \Bigg).
\end{align*}
\normalsize
Setting up a CLIME estimator on $\{\hat{\BS v}_t\}$ leads to $\hat {\BS \Sigma}_v^{-1,CLIME}$ and following now the arguments of \cite{cai2011constrained} gives us that the CLIME estimator fulfill $\|{\BS \Sigma}_v^{-1}-\hat {\BS \Sigma}_v^{-1,CLIME}\|_l=O_P(k_v (\|{\BS \Sigma}_v^{-1}\|_1 \|{\BS \Sigma}_v- \hat {\BS \Sigma}_v\|_{\max}^{1-q_v}))$ for $l\in [1,\infty]$.

We have by Theorem~2 that $\|\A-\hat{\A}\|_{\max}=O_P( \max_s\|\hat \bbeta^{(s)}- \bbeta^{(s)}\|_2)$. Consequently, we obtain by Theorem 1 in \cite{krampe2020statistical} that under Assumption~\ref{sparsity.b} $\sum_{j=1}^p \|\hat{\BS A}^{(thr,j)}-\BS A^{(j)}\|_l= O(k \max_s\|\hat \bbeta^{(s)}- \bbeta^{(s)}\|_2^{1-q})$. Then, we have by Theorem~6 in \cite{krampe2020statistical}
$\|\BS f_\xi(\omega)^{-1}-\hat{\BS f}_\xi(\omega)^{-1}\|_l=O_P(\sum_{j=1}^p\|\BS A^{(j)}\|_l^2 \|{\BS \Sigma}_v^{-1}-\hat {\BS \Sigma}_v^{-1,CLIME}\|_l+\sum_{j=1}^p \|\hat{\BS A}^{(thr,j)}-\BS A^{(j)}\|_l \|\BS A^{(j)}\|_l \|\BS \Sigma_v\|_l)$. 
\end{proof}
\end{document}